\documentclass[opre,nonblindrev]{informs3}
\usepackage{multicol}

\usepackage{mathtools}
\OneAndAHalfSpacedXI

\newcommand{\abs}[1]{\left\vert{#1}\right\vert}

\newcommand{\ind}{\mathbb{I}}
\renewcommand{\P}{\mathbb{P}}
\newcommand{\E}{\mathbb{E}}

\def \loc{\text{Loc}}

\def\m{m}
\def\hu{\text{\sf H}}

\def\bp{{{\mathbf{p}}}}
\def\bs{{{\mathbf{s}}}}

\def\vminus{{{\mathbf{v}}_{-i}}}
\newcommand{\Q}{\mathtt{q}}
\newcommand{\p}{\mathtt{r}}

\def\R{\small{\sf R}}

\def \Mat{\cal M}
\def \mech{\mathbf M}

\newcommand{\PP}{\mathsf{SPP}}
\newcommand{\PPH}{\PP_\hu}
\newcommand{\PPM}{\PP_{\Mat}}
\newcommand{\EG}{\mathsf{ESP}}

\newcommand{\Opt}{\mathsf{Opt}}

\newcommand{\OptM}{\Opt_\Mat}
\newcommand{\feas}{\cal F}

\newcommand{\MP}{\mathsf{MP}}
\newcommand{\T}{{\cal T}}

\newcommand{\ME}{\mathsf{ME}}
\newcommand{\UE}{\mathsf{UE}}
\newcommand{\UP}{\mathsf{UP}}

\usepackage{array}

\makeatletter
\newcommand{\thickhline}{%
    \noalign {\ifnum 0=`}\fi \hrule height 1pt
    \futurelet \reserved@a \@xhline
}
\newcolumntype{"}{@{\hskip\tabcolsep\vrule width 1pt\hskip\tabcolsep}}
\makeatother
\usepackage{multirow}
\usepackage{amssymb, amsmath}
\usepackage{bbm}
\usepackage{upgreek}
\usepackage{array}
\usepackage[cal=boondox]{mathalfa}
\def\makeenmark{\hbox to1.275em{\theenmark.\enskip\hss}}
\def\enoteformat{\rightskip0pt\leftskip0pt\parindent=1.275em
  \leavevmode\llap{\makeenmark}}
  
\makeatletter
\newenvironment{oneshot}[1]{\@begintheorem{#1}{\unskip}}{\@endtheorem}
\makeatother
\usepackage{natbib}
 \bibpunct[, ]{(}{)}{,}{a}{}{,}%
 \def\bibfont{\small}%

\TheoremsNumberedThrough     

\EquationsNumberedThrough    
\begin{document}
\RUNAUTHOR{Beyhaghi et al.}

\RUNTITLE{Improved Revenue Bounds for Posted-Price and Second-Price Mechanisms}
\TITLE{{Improved Revenue Bounds for Posted-Price and Second-Price Mechanisms}}

\ARTICLEAUTHORS{%
\AUTHOR{Hedyeh Beyhaghi,}
\AFF{Toyota Technological Institute at Chicago, Chicago, IL,\EMAIL{
hedyeh@ttic.edu}} 
\AUTHOR{Negin Golrezaei,}
\AFF{Sloan School of Management, Massachusetts Institute of Technology, Cambridge, MA, \EMAIL{golrezae@mit.edu}}
\AUTHOR{Renato Paes Leme, Martin P\'{a}l, and Balasubramanian Sivan}
\AFF{Google Research, New York, NY, \EMAIL{renatoppl, mpal, balusivan@google.com}}
}

\ABSTRACT{
	We study revenue maximization through sequential posted-price (SPP) mechanisms in single-dimensional settings with $n$ buyers and independent but not necessarily identical value distributions. We construct the SPP mechanisms by considering the best of two simple pricing rules: one that imitates the revenue optimal mchanism, namely the Myersonian mechanism, via the taxation principle and the other that posts a uniform price. Our pricing rules are rather generalizable and yield the first improvement over long-established approximation factors in several settings. We design factor-revealing mathematical programs that crisply capture the approximation factor of our SPP mechanism. In the single-unit setting, our SPP mechanism yields a better approximation factor than the state of the art prior to our work~\citep{azar2017prophet}. 
	In the multi-unit setting, our SPP mechanism yields the first improved approximation factor over the state of the art after over nine years~(\cite{yan2011mechanism} and \cite{TEGMM10}). Our results on SPP mechanisms immediately imply improved performance guarantees for the equivalent free-order prophet inequality problem. In the position auction setting, our SPP mechanism yields the first higher-than $1-1/e$ approximation factor. In eager second-price (ESP) auctions, our two simple pricing rules lead to the first improved approximation factor that is strictly greater than what is obtained by the SPP mechanism in the single-unit setting. 
}


\maketitle
\section{Introduction}
\label{sec:intro}
{\color{black}The seminal work of \cite{myerson1981optimal} and a generalization by~\cite{archer2001truthful} established the revenue optimal mechanism in general single-dimensional settings. }This mechanism is optimal among all possible Bayesian incentive-compatible and interim individually rational mechanisms. This optimal mechanism, which is also referred to as the Myersonian mechanism, is a wonderful conceptual vehicle, but it is rarely used in practice due to its complex structure and strong dependence on buyer value distributions. {\color{black}In the Myersonian mechanism, how the auctioned item is allocated and what buyers pay are considerably complex to communicate to the buyers.} When one is forced to run a more natural but sub-optimal mechanism like sequential posted-price (SPP) mechanisms, one of the primary questions of interest is to lower bound the fraction of the optimal revenue that SPP mechanisms can obtain. In this paper, we develop new structural insights into the design of sequential posted-price (SPP) mechanisms to establish improved revenue approximation factors with respect to the optimal mechanism in single-dimensional settings

 In a general single-dimensional setting, there are $n$ buyers with independent but potentially non-identical value distributions as well as a feasibility constraint on which set of buyers can be simultaneously served. SPP mechanisms compute one price per buyer and approach buyers in the descending order of prices, making take-it-or-leave-it offers at the posted price. {\color{black}Running SPP mechanisms to determine allocation and payment satisfies numerous desired properties, including trivial game dynamics for the buyers, buyers not having to reveal their private values, sellers not having to assemble all buyers together to decide allocation/payment, etc. We refer the reader to~ \cite{chawla2010multi} for a more detailed discussion.  These properties make SPP mechanisms objects of both practical relevance and scientific interest.}

The focal point of our work is the development of two \emph{pricing rules} that compute the prices to be posted to buyers in {\color{black}an SPP} mechanism. {\color{black}We name our two pricing rules  \emph{Myersonian pricing rule}  {\color{black}(which we also refer to as the Myersonian SPP mechanism)} and the \emph{uniform pricing rule}.}
Our final SPP mechanism sets prices using one of them, depending on which one offers higher expected revenue for the distribution in hand. The Myersonian pricing rule, inspired by the Myersonian mechanism, sets the {\color{black} taxation-principle-prescribed} prices for each buyer. The taxation principle says that any deterministic incentive-compatible mechanism (which includes the Myersonian Mechanism, see Section~\ref{sec:taxation}) can be interpreted as {\color{black}an SPP} mechanism, except that the posted price for each buyer is a function of other buyers' values. Of course, this is not really {\color{black}an SPP} mechanism as the latter does not allow us to solicit buyer values. The twist in our Myersonian pricing rule is that for each buyer $i$, we \emph{sample} the values of other buyers $j \neq i$ from their distributions, and compute 
the posted price for $i$ that the {\color{black}taxation-principle} interpretation of the Myersonian mechanism would have yielded. We perform {\color{black}{fresh and independent sampling}} while computing the prices for 
different buyers. Thus, while the prices in the Myersonian mechanism are highly correlated across buyers, they are independent in the Myersonian pricing rule, crucially helping our analysis. 
The \emph{uniform} pricing rule posts a single (anonymous) price across all buyers.

{\color{black}The design of our two pricing rules is motivated by our observation that in the worst-case example for the Myersonian pricing rule, the uniform pricing rule performs almost optimally. This suggests that these two pricing rule can complement each other.} Building on this conceptual understanding, we write a novel factor-revealing mathematical program whose objective captures the optimal revenue and whose constraints enforce the revenue of the two aforementioned SPP mechanisms to be at most $1$. The resulting mathematical program lends itself to a clean solution, with the optimal objective value directly yielding the approximation factor. We apply these two pricing rules and our factor-revealing technique to many settings and obtain improved approximation factors. {\color{black}\mbox{Table \ref{table:results}}} summarizes our most important results in the different settings that we study and compares these to the best-known bounds prior to this work. We now discuss our results in detail.

\begin{table}
\footnotesize
\begin{center}
\makebox[\textwidth][c]{ 
\begin{tabular}{ |c|| c|c||cc|}\hline
  \multirow{2}{*}{Setting}      & \multicolumn{1}{c|}{Our Universal}                  &    Prior Universal  & \multicolumn{1}{c|} {Our n-dependent}     & Prior  n-dependent    \\
					    & \multicolumn{1}{c|}{Bound }                        & Bound                                       & \multicolumn{1}{c|} {Bound}  & Bound             \\ \hline
\multirow{2}{*}{$1$-unit SPP} & \multirow{2}{*}{$0.6543$}  &   \multirow{4}{*}{ \begin{tabular}{c}$ 0.6346$\\[0.3em]
                                      \citep{azar2017prophet}
                \end{tabular}} &  \multicolumn{1}{c|} {\multirow{2}{*}{\begin{tabular}{c}{See 3rd row of}\\[0.3em]
                                   Table \ref{tab:smalln}                \end{tabular}} }               
                 & \multirow{4}{*}{ \begin{tabular}{c}$0.6346$ for $n \ge 74$,\\
                                     \citep{azar2017prophet}\\
                                     $1-(1-\frac{1}{n})^n$   for $n < 74$\\
                                     \citep{chawla2010multi}
                \end{tabular}}

\\
& && \multicolumn{1}{c|} {} & 
   \\\cline{1-2}\cline{4-4}
   \multirow{2}{*}{$1$-unit ESP}& \multirow{2}{*}{$0.6620$} & &  \multicolumn{1}{c|} {See 4th row of } &        \\
   & &&\multicolumn{1}{c|} {Table \ref{tab:smalln}}&    \\\hline
    \multirow{2}{*} {$\hu$-unit SPP}  & \multirow{2}{*}{\begin{tabular}{c}{See }\\ 
                                   Table \ref{table:multi}                \end{tabular}}   & \multirow{2}{*}{$1- \frac{\hu^\hu}{\hu! e^{\hu}}$ 
                                   \citep{yan2011mechanism} }                               
                                                                       &  \multicolumn{2}{c|}{\multirow{4}{*}{-}}
                                 \\
       & & & &  \\ \cline{1-3}
     \multirow{2}{*} {Position Auctions} & \multirow{2}{*}{0.6543} &\multirow{2}{*}{-} &&
   \\
    & &&&\\ \hline
       \end{tabular}
  } 
\end{center}
	\caption{The bounds that we achieve, compared to the best-known bounds prior to this work. Followup work by~\cite{CSZ19} improves the best approximation factor for $1$-unit SPP and $1$-unit ESP to $0.669$.}
\label{table:results}
\end{table}
\textbf{SPP in Single-unit Settings.} The first approximation factor for SPP mechanisms in single-unit ($1$-unit) settings was established by ~\citet{chawla2010multi}, who show that SPP mechanisms obtain a $1-\frac{1}{e}$ fraction of the optimal revenue---that is, a $1-\frac{1}{e}$ approximation factor. Since the result of \cite{chawla2010multi}, the same $1-\frac{1}{e}$ approximation \big (or, more generally, the $1-(1-\frac{1}{n})^n$ approximation, where $n$ is the number of buyers\big), for SPP mechanisms was found to be obtainable with various techniques, including pipage rounding \citep{calinescu2011maximizing} and  correlation gap (\cite{agrawal2012price} and~\cite{yan2011mechanism}). The first improvement over this $1-(1-\frac{1}{n})^n$  bound was achieved recently by~\cite{azar2017prophet}, who show how to improve $1-(1-\frac{1}{n})^n$ to $1-\frac{1}{e} + 1/400 \approx 0.6346$ for $n \geq 74$, while still staying at $1-(1-\frac{1}{n})^n$ for  $n<74$. The $n$-dependent bounds obtained by~\cite{azar2017prophet} leads to a universal bound (i.e., valid 
for any number of buyers $n$)  of  $0.6346$, which was the best universal bound prior to this work.

Our pricing rules and factor-revealing technique enable us to provide an improved universal bound of $0.6543$ for SPP mechanisms in the single-unit setting. Note that the improvement from the universal bound of $0.6346$ by \cite{azar2017prophet} to $0.6543$  for SPP mechanisms is significant in light of the fact that SPP mechanisms cannot yield more than a $0.745$ fraction of the optimal revenue even when the valuations are i.i.d. (see~\cite{HK82} and also~\cite{correa2017posted}).\footnote{The results in \cite{HK82} and \cite{correa2017posted} are presented for the equivalent prophet inequalities setting. We discuss the prophet inequalities setting in Section \ref{sec:freeorderprophet}.} The worst-case universal bound of $0.6543$ (worst-case occurs when $n$ goes 
to infinity), is useful in settings where either there is significant uncertainty in the number of buyers or the number of buyers is rather large. For smaller $n$, our approximation factors are noticeably larger. Table \ref{tab:smalln} presents our improved $n$-dependent bounds for {\color{black}$n \in \{1, 2, \ldots, 10\}$} along with the value of $1-(1-\frac{1}{n})^n$ for comparison (the last row in the table for ESP will be explained later). 
  
  \begin{table}[h!]
\begin{center}\setlength\extrarowheight{1pt}
\footnotesize{
\begin{tabular}{ |c||c|c|c|c|c|c|c|c|c|c| }
\hline 
   $n$ & 1&2&3&4&5&6&7&8&9&10 \\
 \hline
 $1-(1-1/n)^n$ &$1.0000$&$0.7500$&$0.7037$&$0.6836$ &$0.6723$& $0.6651$ & $0.6601$&$0.6564$&$0.6536$ &$0.6513$\\ \hline
 $1$-unit SPP& $1.0000$  &$0.7586$&$0.7168$&$0.6990$&$0.6891$&$0.6828$&$0.6785$& $0.6753$&$0.6728$&$0.6709$\\ \hline
ESP &$1.0000$&$0.7611$&$0.7210$&$0.7040$&$0.6946$&$0.6887$&$0.6846$&$0.6815$&$0.6792$&$0.6774$\\\hline
\end{tabular}}
\end{center}
\vspace{0.1in}
  \caption{Approximation factors of SPP mechanisms and ESP auctions for different numbers of buyers $n$.}
\label{tab:smalln}
\end{table}

\textbf{SPP in $\hu$-unit Settings.} For the $\hu$-unit (multi-unit) setting, we beat the correlation-gap-generated factor of $1-\frac{\hu^{\hu}}{\hu!e^{\hu}}$ by \cite{yan2011mechanism} (the same bound as~\citeauthor{yan2011mechanism} was obtained in the independent work by~\cite{TEGMM10} without the correlation-gap machinery). The exact factor we obtain for different values of $\hu$ is provided in Table~\ref{table:multi}. Beating the known factor necessarily requires a deeper understanding of SPP mechanisms than using a black-box hammer like correlation gap. We obtain such {\color{black}understanding} via our two simple pricing rules. {\color{black}Analyzing our pricing rules in the $\hu$-unit setting, which is done by relating the variables of the Myersonian mechanism to  the variables of the two SPP mechanisms in a factor-revealing mathematical program,  is  technically challenging.} 
Overcoming this challenge yields some neat combinatorial lemmas. See, for example, Lemma~\ref{lm:equal}, where we relate the revenue of the Myersonian SPP mechanism to the optimal revenue.

{\color{black}
 \begin{table}[h]
\setlength{\extrarowheight}{4pt}
\centering
\fontsize{9}{9}\selectfont{{
\begin{tabular}{|c |c |c |c|c|c|c|c|c|c|c|c|c|} 
 \hline
 $\hu$ & 1 & 2 & 3 & 4 & 5 & 6 &7& 8 & 9 & 10\\ 
  \hline
 $1- \frac{\hu^\hu}{\hu! e^{\hu}}$ & 0.6321   & 0.7293  &  0.7760   & 0.8046 &   0.8245 &   0.8394  &  0.8510   & 0.8604 &   0.8682   & 0.8749\\ \hline
 {\color{black}Our Bound for the $\hu$-unit SPP}  & 0.6543   & 0.7427 &   0.7857   & 0.8125 &   0.8311   & 0.8454  &  0.8567   & 0.8656   & 0.8734   & 0.8807\\  \hline
\end{tabular}
}}
\vspace{1em} 
\caption{{The second row presents the best-known bound for the multi-unit setting prior to this work, and the third row presents our improved bound for $\hu\in\{1, 2,\ldots, 10\}$.} Approximation factors are applicable for all values of $n$.}
\label{table:multi}
\end{table}}

\textbf{SPP in Matroidal Settings.} We show that our improved bounds for the multi-unit settings leads to identical improved bounds for partition matroid settings (see Section~\ref{sec:partitionmatroid}). For a general matroid setting, our pricing rules yield an alternate set of prices to achieve the $1-\frac{1}{e}$ approximation from~\cite{yan2011mechanism} (see Section~\ref{sec:matroid}). We leave beating the $1-\frac{1}{e}$ factor for general matroids as an open question and believe that the techniques from this study could be of use in doing this. 

\textbf{SPP in Position Auction Settings.} Position auctions are ubiquitous in search advertising markets (\cite{edelman2007internet,varian2007position,ostrovsky2011reserve,athey2011position} and \cite{lucier2012revenue}). In a position auction setting, there are $n$ buyers (advertisers) and $n$ positions with different click-through rates, and the goal is to assign buyers to the positions. For the position auction setting, even a $1-\frac{1}{e}$ approximation was not known prior to this work, and we obtained a strictly larger than $1-\frac{1}{e}$ approximation factor of $0.6543$. In order to apply our technique to this setting, we show that the optimal position auction can be described as a combination of $n$ multi-unit auctions. This enables us to take advantage of our SPP mechanisms for multi-unit settings to design a novel SPP mechanism with $n^2$ prices {(a price for each buyer and position)}. We show that our SPP mechanism for the position auctions obtains a universal bound of $0.6543$. We further obtain improved bounds for a given vector of click-through rates; see Theorem \ref{thm:PA}.

\textbf{Eager Second-Price Auctions.}
As a further demonstration of the generalizability of our technique, we analyze eager\footnote{There are two different ways that personalized reserve prices can be applied in second-price auctions: lazy and eager \citep{dhangwatnotai2015revenue}. In the lazy version, we first determine the potential	winner and then apply the reserve prices. In the eager version, we first
	apply the reserve prices and then determine the winner. See Section \ref{sec:eager} for details.} second-price (ESP) auctions using our two pricing rules and factor-revealing technique. While~\citet{chawla2010multi} show that an ESP that uses the prices yielded by an SPP mechanism as reserve prices always obtains a weakly higher revenue than the SPP mechanism, no technique has been known to provide 
strictly better approximation factors for ESP than SPP. We use our factor-revealing technique to achieve this. Our universal improved bound for ESP in the $1$-unit setting is $0.6620$, which is strictly greater than that for the SPP mechanism in $1$-unit setting (i.e., $0.6543$). Our $n$-dependent bounds in the $1$-unit setting for ESP, which are presented in the last row of Table \ref{tab:smalln}, are also strictly greater than $n$-dependent bounds for the SPP mechanism. We note the best-known bound for ESP prior to this work was $0.6346$ by \cite{azar2017prophet}. Our improvement from $0.6346$ to $0.6620$ is significant because (i) the ESP auctions cannot obtain a fraction of the optimal revenue that is greater than $0.778$ (see~\cite{MS19}), and (ii) these auctions are run \emph{several billions} of times each day by ad exchanges to allocate ad slots.

\textbf{Connection to Free-order Prophet Inequalities.} Our improved bounds for the SPP mechanism in the single-unit, multi-unit, and partition matroid settings directly imply an identical improvement in the free-order prophet inequality problem in the corresponding settings due to the recent equivalence established by~\cite{correa2017back}. More details on this connection are presented in the Related Work section (Section~\ref{sec:related}).

\paragraph{Organization:} 
We discuss related work in Section~\ref{sec:related}.
Section~\ref{sec:model} formally introduces the model. Section~\ref{sec:single-unit} discusses the single-unit case, and Section~\ref{sec:multi-unit} presents our bounds for the $\hu$-unit setting.  Section~\ref{sec:extend} uses the $\hu$-unit result to derive results for the position auction and partition matroid environments and it also briefly discusses the general matroid setting. Our improved approximation factors for ESP auctions are presented in Section~\ref{sec:eager}.

\section{Related Work}
\label{sec:related}
{\color{black} Our work contributes to the literature on approximating the optimal auction with simple auction formats.}
As stated earlier, the structure of 
 the optimal mechanism can be complex when the value distributions are irregular and/or non-identical across buyers. Because of this, several papers have studied simpler auction formats, such as second-price auctions with (personalized) reserve prices (\cite{HR09, paes2016field, RW16, allouah2018prior}, and \cite{derakhshan2019lp}), boosted second-price auctions \citep{golrezaei2017boosted},  buy-it-now or take-a-chance (BIN-TAC) 
mechanisms \citep{celis2014buy}, and first-price auctions \citep{bhalgat2012online, balseiro2019contextual}, to name a few.

{\citet{HR09} study the question of approximating the optimal revenue via a second-price auction with personalized reserve prices {\color{black} when buyers' valuation distributions  are independent but non-identical}.} They show that for regular distributions the second-price auction with so-called monopoly reserve prices yields a $2$-approximation; however, for irregular distributions, no constant factor approximation is possible.~\citet{paes2016field} consider second-price auctions and study the question of computing the optimal personalized reserve prices in a correlated distribution setting. Further, they show that the problem is NP-complete. ~\citet{RW16} indicate that this problem is 
APX-hard for correlated distributions and provide a $\frac{1}{2}$-approximation. An improved approximation of $0.684$ is subsequently obtained by \cite{derakhshan2019lp}. We note that in a correlated distribution setting, the benchmark is not the optimal revenue; instead, it is the maximum revenue that second-price auctions with optimal reserve prices can obtain. {In the current study, we provide an improved approximation factor for eager second-price auctions in an independent distribution setting and {show that this auction format---despite its simple structure---performs well, even when the distributions are heterogeneous and irregular.} }

\subsection{Free-order Prophet Inequalities}\label{sec:freeorderprophet} As mentioned earlier, our improved bounds for the SPP mechanism in the single-unit, multi-unit, and partition matroid settings directly imply an identical improvement in the free-order prophet inequality problem in the corresponding settings. In the free-order prophet inequality problem, there are $n$ independent random variables
with known distributions. Upon inspecting a variable, a gambler learns its realized value and must choose between stopping and obtaining its value as a reward or abandoning that variable forever and continuing to inspect other variables. The gambler can choose the order with which he wants to inspect the variables. His goal is to maximize his reward by competing with a prophet that knows all the realized values of the variables. The set of variables that can be feasibly selected can be from any feasibility constraint (like single-unit, multi-unit, matroids, etc.). Recently, ~\cite{correa2017back} showed that any approximation factor for SPP mechanisms directly translates to the same approximation factor for free-order prophet inequalities in numerous environments, including matroid feasibility constraints. Thus, our first improvements in various settings directly imply a first improvement in the corresponding prophet inequality problem as well.

\textbf{Related Work on Prophet Inequalities.} The literature on prophet inequalities is vast (\cite{KS77,KS78}). Here, we provide a quick overview, focusing on the case where not more than one random variable can be selected (single-unit). There are three variants that are commonly studied: adversarial-order prophets, free-order prophets, and random-order prophets. In the free-order setting, the gambler can select the order of the random variables that he inspects. As stated earlier, due to the results of~\cite{correa2017back}, our SPP mechanism bound improvements in the $1$-unit and $\hu$-unit settings immediately yield the same improved bounds for free-order prophet settings. In the adversarial-order prophets version, the order is determined by an adversary. For this version,~\citet{KS78} show that when variables are independent but not necessarily identical, the gambler can obtain at least $\frac{1}{2}$ of the expected value obtained by a prophet; subsequently, ~\cite{samuel-cahn1984} showed the same with a single-threshold policy. In the random-order prophets version, the random variables arrive in a uniformly random order; the results for this version are discussed in the paragraph on subsequent work below. When the random variables are i.i.d., all the three variants (free, random, and adversarial) coincide. ~\citet{HK82} show that the gambler can obtain at least $1-\frac{1}{e}$ of the prophet's value and they also show examples that prove that one cannot obtain a factor beyond $\frac{1}{1.342}\sim 0.745$.~\cite{correa2017posted} show a matching $0.745$ approximation for the i.i.d. version. 
We highlight that this $0.745$ result is not applicable to our setting, as in our study, the buyer valuations are not i.i.d.

\textbf{Subsequent Work.} After the appearance of the first version of this paper, for the $1$-unit setting, ~\cite{CSZ19} obtain improved approximation factors for both the free-order prophet and random-order prophet problems, obtaining an approximation factor of  $0.669$. Further, they show that for the random-order prophets problem, it is not possible to obtain an approximation of over $\sqrt{3}-1~=0.732$, thereby separating it from the approximation of $0.745$ that was obtained for the i.i.d. case (see~\cite{correa2017posted}). We note here that while this approximation factor of $0.669$ for SPP mechanisms beats our $0.6543$ factor, {\color{black}the improved bound of \cite{CSZ19}, is restricted to the $1$-unit setting.} The only known results that beat the long-established approximation factors when selling more than one unit are the results in our paper. In addition, even in the $1$-unit setting, our $n$-dependent bounds are strictly greater than $0.669$ when $n\le 10$; see Table \ref{tab:smalln}.

\textbf{Posted Prices, Prophet Inequalities, and Generalizations.} 
The establishment of a connection between prophet inequalities and mechanism design was initiated by~\citet{HKS07}, who interpreted the prophet inequality algorithms as truthful mechanisms for online auctions. ~\cite{CHK07} obtain an approximation of $4$ for the single-buyer $n$-items unit-demand pricing problem by upper bounding the revenue of the multiparameter setting by that of the single-unit $n$-buyer single-parameter problem. \citet{chawla2010multi,chawla2010randomized} expand this connection and develop constant fraction approximations for several multiparameter unit-demand settings by establishing constant factor approximations to the corresponding single-parameter posted-price settings through connections to prophet inequalities. ~\citet{yan2011mechanism} makes a connection between the revenue of SPP mechanisms in the $\hu$-unit setting and the correlation gap for submodular functions~\citep{agrawal2012price}. For the  $\hu$-unit setting, as stated earlier, \citet{TEGMM10} establish the same bound in \citet{yan2011mechanism} without using the correlation-gap machinery. They further develop a PTAS for computing the optimal \emph{adaptive} SPP mechanisms in an $\hu$-unit single-parameter setting when $\hu$ reaches infinity. Recall that in our SPP mechanisms, the prices are not adaptive.

\section{Model}\label{sec:model}
{In this paper, we study the single-unit, multi-unit, matroidal, and position auction settings. In this section, we describe the mechanisms for the matroidal setting, which includes single-unit and multi-unit settings. The description of position auctions is provided when these results are discussed.} 

{\textbf{Buyers' Values.} There are $n$ buyers indexed by $i \in [n]$, where $[n] =\{1, 2, \ldots, n\}$. All the buyers demand, at most, one unit of the item. Buyer $i$ has a private value $v_i$ for receiving one unit of the item being sold, where $v_i$ is drawn independently from a publicly known distribution $F_i$. The value distributions are either continuous probability measures (with no atoms) or discrete probability measures with finite support. This minor restriction on distributions is for technical reasons and we comment on how this restriction is used toward the end of Section~\ref{sec:taxation}.}

{\textbf{General Feasibility Constraints.} Let $\feas$ be an arbitrary collection of subsets of $[n]$. We say that a mechanism has a feasibility constraint $\feas$  if the set of all buyers that simultaneously receive an allocation in the mechanism has to be a set in $\feas$. The feasibility constraints that we study are:}

\begin{itemize}
	\item { \textbf{Single-unit Setting.} $\feas$ is the collection of $n$ singleton sets and the empty set. Such a feasibility constraint is generated by the mechanism where there is only a single item to sell and, therefore, the set of allocated buyers must either be a singleton set or an empty one. We also refer to this setting as a $1$-unit setting.}
	\item {\textbf{Multi-unit Setting.} $\feas$ is the collection of subsets of $[n]$ of size at most $\hu\ge 1$. Such a feasibility constraint is generated by the mechanism having just $\hu$ units of an item to sell. We also refer to this setting as an $\hu$-unit. Note that the $\hu$-unit setting subsumes the $1$-unit setting.}
	\item {\textbf{Matroidal 	Setting.} $\feas$ is the collection of subsets of $[n]$ that include all the independent sets of a matroid. A matroid $\Mat(E, I)$ comprises a ground set of elements $E$ and a non-empty collection 	$I\subseteq 2^E$ of independent sets. A matroid $\Mat(E, I)$ satisfies the following two conditions. (i) If set $T\in I$, then any subset of $T$ is also in $I$. (ii) Given $S, T\in I$ with $|S|< |T|$ and element $e\in T\backslash S$, we have $S\cup \{e\} \in I$. 
	For our purposes, the ground set is the set of buyers, $E =[n]$, and $I = \feas$---that is, $I$ consists of all the feasible allocations. In particular, $T\notin I$ implies that the set, $T$, of buyers cannot be allocated simultaneously. Note that the single-unit and multi-unit settings are special cases of the matroidal feasibility constraint.}
\end{itemize}

\medskip
\textbf{Sequential Posted-Price Mechanisms $\PPM(\bp)$.}
For matroidal settings, the SPP mechanism is denoted by
$\PPM(\bp)$. { \color{black} Here, $\bp =(p_1, p_2, \ldots, p_n)$ is a vector of posted prices in $\PPM(\bp)$. Without loss of generality, we assume that $p_1\ge p_2\ge \ldots\ge p_n$.}\footnote{\color{black} When we re-order the prices to make them decreasing, we re-index the buyers accordingly.} The mechanism $\PPM(\bp)$ approaches buyers in decreasing order of the sellers' posted prices. If adding buyer $i$ to the already allocated set $S$ of buyers satisfies the feasibility constraint---that is, if $S \cup \{i\} \in \feas$---the mechanism offers price $p_i$ to buyer $i$. 
If the buyer accepts the offer---that is, $v_i\ge p_i$---buyer $i$ will be allocated to and charged a price of $p_i$, and $S$ will be updated to $S \cup \{i\}$. Otherwise, the mechanism proceeds to buyer $i+1$.

{ \color{black}For the special case of $\hu$-unit setting, we denote  $\PPM(\bp)$ by $\PPH(\bp)$.
When $\hu=1$, i.e., for the $1$-unit setting, we exclude the subscript $\hu$. Throughout the proof, with a slight abuse of notation, we denote the expected revenue of $\PPM(\bp)$ by $\PPM(\bp)$, where the expectation is with respect to (w.r.t.) the randomness in the buyers' values.}

\subsection{Optimal Revenue Benchmark}
\label{sec:taxation}
{Our benchmark, which we refer to as $\OptM$ (we exclude the subscript when it is evident from the context), is an incentive-compatible (truthful) and individually rational revenue-optimal
auction in the independent value setting. This {mechanism} was designed by \citet{myerson1981optimal}. For most of our results, the specific form of the optimal mechanism is irrelevant. We use only the fact that it is a deterministic truthful mechanism. Hence, the taxation principle provides a simple equivalent form of expressing such a mechanism; {\color{black}see the discussion after Theorem 2 in \cite{hammond1979straightforward}.} {Note that even when the distributions are not regular, the Myersonian mechanism can be {implemented} as a deterministic mechanism; see \cite{myerson1981optimal} and \cite{chawla2014bayesian}.}}

{The following lemma describes the taxation principle in any deterministic truthful mechanism. We do not prove it here, as the proof for this can be derived from any standard auction theory textbook; see, for example, Theorem 9.36 in \cite{NRTV07}.

{\begin{lemma}[Taxation {Principle}]\label{taxation}
Given a deterministic truthful mechanism $\mech$ in any one-dimensional independent private-value setting with arbitrary feasibility constraints (including single-unit, multi-unit, and matroids), there
are threshold functions $t_i(\vminus)$ for each buyer $i${,} that depend only on the bids of other buyers $\vminus = (v_j)_{j \neq i}${,} such that the allocation and payment of mechanism $\mech$ can be described in the following manner:
\begin{itemize}
\item if $v_i > t_i(\vminus)$, then  buyer $i$ is allocated and he is charged the
threshold $t_i(\vminus)$.
\item if $v_i = t_i(\vminus)$, then either  buyer $i$ is allocated and charged $t_i(\vminus)$ or is not allocated and not charged. 

\item if $v_i < t_i(\vminus)$, then  buyer $i$ is not allocated and not charged.
\end{itemize}
\end{lemma}}}

\medskip

{We note that the threshold function $t_i$ can be computed for any deterministic incentive-compatible mechanism. In such a mechanism, each buyer $i$ has a critical value $v_{\min}$---which depends only  on other buyers' values---such that he is allocated and pays $v_{\min}$ if his value $v_i > v_{\min}$. When $v_i < v_{\min}$, he does not get allocated and pays $0$. When $v_i = v_{\min}$, the mechanism can either allocate an item to buyer $i$ and charge him $v_{\min}$ or not allocate any item to him and charge $0$. The threshold function is defined as $t_i(\vminus) = v_{\min}$. We provide two examples  to illustrate how these thresholds are computed.} 

{\begin{example} [Second-price auction with $1$-unit and no reserve price] Here, the critical value for buyer $i$ is $t_i(\vminus) = \max_{j \neq i} v_j$, as in the second-price auctions with no reserve price, a buyer with the highest value/bid is allocated. 
\end{example}}

{\begin{example} [Optimal mechanism with $1$-unit and uniform distributions] Suppose that there are two buyers with values $v_1$ and $v_2$ drawn independently from the uniform distributions in [0,1] and [0,2], respectively. In the optimal mechanism, the item is allocated to the buyer with the highest non-negative virtual value.\footnote{The virtual value of buyer $i$ with value $v\sim F_i$ is given by $v- \frac{1-F_i(v)}{f_i(v)}$.} The virtual values of buyers $1$ and $2$ are $2v_1-1$ and $2v_2-2$, respectively. Thus, buyer $1$ is allocated when $v_1 \ge \max(0.5, v_2 - 0.5)$, and buyer $2$ is allocated when $v_2 \ge  \max(1, v_1 + 0.5)$. Therefore, the threshold functions are given by $t_1(v_2) = \max(0.5, v_2 - 0.5)$ and $t_2(v_1) = \max(1, v_1 + 0.5)$.
\end{example}}

Thresholds {$t_i(\vminus)$} are constructed in such a manner that the set of buyers with value strictly above thresholds can always be simultaneously allocated without violating any feasibility constraints (for example, in the $\hu$-unit setting, at most $\hu$ buyers strictly exceed their thresholds). However, it is possible that the set of buyers who weakly exceed their threshold cannot all be simultaneously allocated (for example, in the $\hu$-unit setting, more than $\hu$ buyers could weakly exceed their thresholds). In such a case, a tie-breaking rule must be determined to accurately describe the allocation rule of the mechanism.

{There are two important cases in which the issue of tie-breaking can be ignored. The first case is when the value distributions are independent and continuous, with no atoms: in this case, the probability that $v_i = t_i(\vminus)$ is zero. The second case is when the distributions have finite discrete support: here, the thresholds for any deterministic mechanism can be constructed\footnote{A deterministic mechanism, if it breaks ties, can only have a deterministic tie-breaking rule. Assume an arbitrary deterministic tie-breaking rule that could break ties differently at different value profiles. We show how to construct thresholds so that the outcome of allocating precisely to the buyers whose value weakly exceeds their threshold is identical to the outcome of the original mechanism. In particular, the set of buyers whose value weakly exceeds their threshold can be feasibly served. Consider the outcome of the original mechanism for buyer $i$ when the values of other buyers are $\vminus$. If $t_i(\vminus)$ is not in the support of  $F_i$ (buyer $i$'s distribution), we set buyer $i$'s threshold at $t_i(\vminus)$. Or, if $i$ is allocated in the original mechanism when his value is $t_i(\vminus)$, we set buyer $i$'s threshold at $t_i(\vminus)$. If buyer $i$ is not allocated in the original mechanism when his value is $t_i(\vminus)$, we set buyer $i$'s threshold at $t_i(\vminus)+\epsilon$, where  $\epsilon > 0$ is such that for any $x$ in the support of $F_i$, where $x > t_i(\vminus)$, we also have $x > t_i(\vminus)+\epsilon$. It now follows that regardless of the true value of buyer $i$, the buyer's allocation in the original mechanism is identical to the one obtained by allocating buyer $i$ exactly whenever his true value weakly exceeds the threshold constructed above.
} in such a way that the set of buyers with value weakly exceeding their threshold---that is, the set of buyers with $v_i \geq t_i(\vminus)$---can always be simultaneously allocated and each allocated buyer pays his threshold. Thus, there is no need to break ties. {In light of this, as stated earlier, our results hold for any continuous probability measures (with no atoms) and discrete probability measures with finite support.
}} 

\subsection{Definitions and Notation} \label{prelim:defn}
\textbf{Thresholds.} In the remainder of the paper, $t_i(\vminus)$ refers to the threshold of buyer $i$ corresponding to the optimal mechanism for the feasibility constraint under study (see 
Lemma~\ref{taxation}). Whenever it is evident from 
the context, we abbreviate $t_i(\vminus)$ or the function $t_i(\cdot)$ by $t_i$. 

\textbf{Re-sampled Thresholds.} We will often  
{refer to the} 
thresholds computed from independently re-sampled values: that is, for 
each buyer $i$, sample $v_{j,i}' \sim F_j$ for all $j \neq i$ and denoted by $t_i(\vminus')$ the re-sampled threshold where $\vminus' = 
(v_{1,i}',\dots,v_{i-1,i}',v_{i+1,i}',\dots,v_{n,i}')$. Observe that we do not reuse samples: for each buyer $i$, we freshly re-sample the values of all other buyers. We abbreviate 
$t_i(\vminus')$ by $t_i'$ whenever it is evident from the context. {Note that although for each $i$, the distribution of $t_i$ is the same as the distribution of $t_i'$, the values of $t_i'$, $i\in [n]$, are independent 
across $i$'s, while the $t_i$'s are correlated.} 

{\textbf{Myersonian Posted Prices.} This refers to the tuple of $n$ (random) posted prices, one per buyer computed from the Myersonian pricing rule. It consists of the re-sampled threshold $t_i'$ for each buyer $i$. Note that the thresholds depend on the feasibility constraint: this is because the optimal mechanism depends on the feasibility constraint $\feas$.}

{\textbf{Uniform Posted Price.} This refers to the highest-revenue-yielding uniform posted price. For the $1$-unit SPP setting, this is given by, $p^{\star} = \argmax_p p\cdot \P[\max_{i\in [n]} v_i \geq p]$. For the $\hu$-unit setting, the uniform price is given by $p^{\star} = \argmax_p p\cdot \E[\min(|S_p({\bf v})|, \hu)]$, where $S_p({\bf v})$ is the set of buyers with $v_i \geq p$.} 

{\textbf{Optimal Revenue.}
Let $s_i^{\star}(\tau ) = \P[v_i~ \geq ~t_i(\vminus) ~\geq ~{\tau}]$, $i\in [n]$, be the probability that {buyer} $i$ wins and pays at least $\tau$ in the optimal mechanism (note that $s_i^{\star}(\tau )$ depends on the feasibility constraints $\feas$), {where the probability is taken w.r.t. $\vminus$ and $v_i$.} Further, define $s^{\star}({\tau}) := \sum_{{i\in [n]}} s_i^{\star}({\tau})$. In the single-unit setting, $s^{\star}({\tau})$ represents the probability that the winner pays at least ${\tau}$. In the multi-unit setting and, more generally, the matroidal setting, $s^{\star}({\tau})$ is the expected number of buyers who receive the item and pay at least ${\tau}$. It follows immediately that $s^{\star}({\tau})$ is a weakly decreasing function whose integral defines the optimal revenue, denoted by $\Opt$:
\begin{align}\textstyle \Opt ~:=~ \int_0^\infty s^{\star}({\tau}) d{\tau}\,. 
\label{eq:opt}\end{align}
Writing the optimal revenue in terms of $s^\star(\cdot)$, which is one of our contributions, enables us to subsequently obtain approximation factors that hold for any regular and irregular value distributions.}

\section{Single-unit SPP Mechanisms}
\label{sec:single-unit}
In this section, we derive a universal bound for the SPP mechanism---that is, the bound holds for any value of the number of buyers $n$. Subsequently, in Section \ref{sec:small_n}, we obtain $n$-dependent bounds using the same principle and techniques as we we use for the universal bound. 
As stated earlier, for the case of $\hu=1$ that we study in this section, we exclude the $\hu$ subscript and denote the mechanism by only $\PP(\cdot)$.
\subsection{A Universal Bound}

The main result of this section is Theorem \ref{thm:pp1}, where we show that in a single-unit $n$-buyer setting with independent private values, there {exists a vector of prices} $\mathtt{\bp} =(\mathtt{p}_1, \mathtt{p}_2, \ldots, \mathtt{p}_n)$, such that $\PP(\mathtt{\bp}) ~\geq~0.6543 \cdot \Opt$, where $\Opt$---as defined in Equation \eqref{eq:opt}---is the expected optimal revenue and $\PP(\mathtt{\bp})$ is the expected revenue of an SPP mechanism with prices $\mathtt{\bp}$. To show this result, we take advantage of the Myersonian and uniform SPP mechanisms. Let $\MP$ denote the expected revenue of the Myersonian SPP mechanism, where the expectation {is taken w.r.t.} the randomness in both the re-sampled posted prices and the buyers' values. Further, let $\UP$ be the expected revenue of the SPP mechanism that posts the best uniform prices, where the expectation is taken w.r.t. the buyers’ values. The proof of Theorem~\ref{thm:pp1} indicates that $\max(\UP, \MP)$ is at least a $0.6543$ 
 fraction of the optimal revenue. Before presenting Theorem~\ref{thm:pp1} and its proof, we illustrate with a simple example how Myersonian and uniform pricing rules complement each other.
  
 \textbf{Myersonian and Uniform Pricing Rules Complement Each Other.} To illustrate how Myersonian and uniform pricing rules complement each other, we show that $\MP\ge (1-1/e)\cdot \Opt$ and present an example in which the Myersonian SPP mechanism obtains exactly a $1-1/e$ fraction of the optimal revenue, while the uniform SPP mechanism is almost optimal there.
{Define $\m({\tau})$ as the probability that the Myersonian SPP mechanism sells with a price of at least ${\tau}$, which is the probability that there is at least one buyer with $v_i ~\geq~ t_i' ~\geq~ {\tau}$. 
In Lemma \ref{lem:prob}, presented at the end of this section, we bound the probability that there is no buyer $i$ with $v_i ~\geq~ t_i' ~\geq~ {\tau}$ as a function of $s^\star(\tau)$; see Inequality (\ref{eq:z_0}). By invoking this inequality, we obtain 
$\m ({\tau}) \geq 
1-\exp{(-s^\star({\tau}))} \ge (1-\exp(-1))s^\star(\tau)$. 
Integrating this expression, we obtain
$$\textstyle \MP = \int_0^\infty m(\tau) d\tau \geq 
(1-e^{-1}) \int_0^\infty s^\star(\tau) d\tau = (1-e^{-1}) \Opt.$$}
Thus far, we have shown that the Myersonian SPP mechanism obtains an approximation factor of $1-1/e$. Next, we present an example that shows that this approximation factor is tight, but a uniform price performs almost optimally, suggesting that Myersonian and uniform pricing rules complement each other.
 Consider the setting where there are $n$ buyers whose values are independently drawn from the uniform distribution in $[1,1+\epsilon]$ for a tiny $\epsilon>0$. Then, the optimal mechanism is simply the second-price auction with a uniform reserve of $1$, and the uniform pricing scheme that posts a price of $1$ comes very close to this optimal mechanism. However, in the Myersonian SPP mechanism, each buyer is offered a random threshold that is the maximum of $n-1$ variables distributed uniformly in $[1,1+\epsilon]$. Thus, each buyer is above such a threshold with probability $1/n$. Since all buyers are independent, with probability $(1-1/n)^n \rightarrow 1/e$, no buyer is above the
threshold. Thus, $\MP$ is merely  a $1-1/e$ approximation for this particular choice of prices. On the other hand, a uniform price of $1$ gets a revenue of $1$ which is very close to optimal. 
  
\medskip
  
\begin{theorem}[Revenue Bound of SPP Mechanisms in Single-unit Settings]\label{thm:pp1}
In a $1$-unit $n$-buyer setting with independent private values, there {exists a vector of prices} $\mathtt{\bp} =(\mathtt{p}_1, \mathtt{p}_2, \ldots, \mathtt{p}_n)$, such that $
	 \PP(\bp) \geq \Opt\cdot\frac{1}{\text{\ref{lp:spm}}} = 0.6543 \cdot \Opt$, where $\Opt$ is the expected optimal revenue in a $1$-unit setting, $\PP(\bp)$ is the expected revenue of the SPP mechanism with prices  $\mathtt{\bp}$, and $\text{\small{\sf{FR}}} = \frac{1}{0.6543}$ is defined as the optimal objective of this factor-revealing mathematical program~\eqref{lp:spm}. 
\begin{align} 
\text{\small{\sf{FR}}} ~=~&{\max_{\{{s(\tau), \tau\ge 0\}}}} ~~\int_0^\infty 
s(\tau) d\tau \nonumber \\
& \begin{aligned}
\text{s.t.} ~~~ &0~ \le~ s(\tau) ~\le~ \min(1, 1/\tau)~~ & &\forall ~~ \tau \ge 0 \nonumber \\
& \int_0^\infty  f(s(\tau)) d\tau ~\le~ 1 \\
&\text{$s(\cdot)$ is weakly decreasing}\,.~~&\nonumber
\end{aligned}
 \label{lp:spm} \tag{{\small{\sf{FR}}}}
\end{align}
 Here,  $f(x) = ({1 - e^{-x}})$.
   \end{theorem}
   
\begin{proof}{Proof of Theorem \ref{thm:pp1}} {We first show that $\max(\MP, \UP) \geq \frac{1}{\ref{lp:spm}}\cdot \Opt$. Subsequently, we prove that $\frac{1}{\ref{lp:spm}} = 0.6543$ (recall that we use \ref{lp:spm} to denote factor-revealing). Without loss of generality, we assume that $\max(\MP, \UP) $ is normalized to one. (This can be done by scaling all values by a constant factor.\footnote{Assume that $\max(\MP,\UP)= w\ne 1$. Then, if one multiplies all the values by a factor of $1/w$, then the optimal revenue $\Opt$, $\MP$ and $\UP$ are all multiplied by the same factor $1/w$ on each value profile. Hence, the ratio of $\max(\MP,\UP)$ to $\Opt$ does not change. In other words, while not all mechanisms are scale-invariant, the three relevant mechanisms for us are all scale-invariant; the argument for this is straightforward based on the definitions of mechanisms.})
{The proof shows that $s^{\star}(\cdot)$ corresponding to the optimal mechanism is a feasible solution to Problem \eqref{lp:spm}. In light of this, the objective function of Problem \eqref{lp:spm} provides an upper bound on the optimal revenue (see Equation (\ref{eq:opt})); hence $\frac{1}{\ref{lp:spm}}$ is the resulting approximation factor.} 

\textbf{Lower Bounds on $\boldsymbol \UP$ (First Set of Constraints).} Consider the SPP mechanism that posts a uniform price of $\tau$ for every buyer. The revenue of this mechanism is equal to $\tau \cdot \P[\max_{i\in [n]} v_i \geq \tau]$, which is at least $\tau s^\star(\tau)$. (Recall that $s_i^{\star}(\tau ) = \P[v_i~ \geq ~t_i ~\geq ~{\tau}]$ and $s^\star(\tau) =\sum_{i\in [n]} s_i^{\star}(\tau)$ is the probability that there is at least one buyer $i$ with $v_i~ \geq ~t_i ~\geq ~{\tau}.$) Therefore, the revenue of the SPP mechanism with the uniform price of $\tau$ is at least $\tau s^\star(\tau)$ for every $\tau \geq0$; that is,
\begin{align}\max_{\tau \ge 0} ~\tau s^\star(\tau) \le \UP\le 1 \,, \label{eq:bound_UP} \end{align}
where the second inequality follows from the fact that the revenue of the SPP mechanism that used the best of Myersonian and uniform prices is normalized to one. {Inequality (\ref{eq:bound_UP})} leads to 
$0 ~\le~ s^\star({\tau})~ \le~ \min(1, 1/{\tau})$ for any $\tau\ge 0,$ confirming $s^\star({\tau})$ satisfies the first set of constraints. Note that here we use the fact that $s^\star(\tau) \leq 1$ for any $\tau \geq 0$---this follows from the two facts that (a) the optimal mechanism is feasible and (b) the thresholds that we construct to mimic the optimal mechanism are such that the set of buyers with values weakly above the threshold can all be simultaneously allocated (for the single-unit case, this implies that no more than one buyer will have a value weakly exceeding the threshold; see Section~\ref{sec:taxation}). To ensure that constructing such feasible thresholds is possible, we assume that the buyers' value distributions are either continuous probability measures (with no atoms) or discrete probability measures with  finite support. 

\textbf{Lower Bounds on $\boldsymbol \MP$ (Second Constraint).} Here, we show that $s^\star (\cdot)$ satisfies the second constraint, {\color{black}i.e.,  $\int_0^\infty  f(s^\star(\tau)) d\tau ~\le~ 1$.} Recall $\m({\tau})$ is the probability that the Myersonian SPP mechanism sells with a price of at least ${\tau}$, which is the probability that there is at least one buyer with $v_i ~\geq~ t_i' ~\geq~ {\tau}$. 
Then, we have
\begin{align}\textstyle \MP ~=~ \int_0^\infty \m ({\tau}) d{\tau} ~\le ~1\,,\label{eq:spm-single}\end{align}
where the inequality follows from $\max(\MP, \UP ) = 1$.  Let 
$Z_{\tau}= \sum_{i=1}^n\ind (v_i 
~\geq~ t_i' ~\geq ~ {\tau})$ be the number of  buyers with $v_i ~\geq~ t_i' ~\geq~ {\tau}$, where $\ind(\cdot)$ is an indicator function ($\ind(A) = 1$ when event $A$ occurs and zero otherwise). As stated earlier, the probability that $Z_{\tau} = 0$ is bounded in Lemma \ref{lem:prob}, which is presented at the end of this section.
By invoking Inequality (\ref{eq:z_0}) in this lemma, we obtain 
\begin{align}\textstyle \m ({\tau}) ~&\ge ~1-\exp{(-s^\star({\tau}))}
     ~=~ f(s^\star(\tau)).\label{eq:u_lower_bound}
\end{align}
{\color{black}Recall that $f(x) ~=~ {1 - \exp(-x)}$}. 
Inequalities (\ref{eq:spm-single}) and (\ref{eq:u_lower_bound})} {\color{black} together imply that $\int_0^\infty  f(s^\star(\tau)) d\tau ~\le~ 1$, which is the desired result.}

\textbf{Solving Problem \eqref{lp:spm}.} Next, we compute the objective value of Problem \eqref{lp:spm}. Define $g(x)= f(x)/ x$. Observe that the second constraint of Problem \eqref{lp:spm} can be written as $\int_0^\infty g(s(\tau)) s(\tau) d\tau ~\le~ 1$. Considering this, it is not difficult to guess the optimal solution of Problem \eqref{lp:spm}.
Since by the last set of constraints of Problem \eqref{lp:spm}, $s(\tau)$ is (weakly) decreasing in $\tau$ and $g(s(\tau))$ is decreasing in $s(\tau)$, we have $g(\tau)$ increasing in $\tau$. Thus, the optimal solution must satisfy that $s(\tau) = \min(1,1/\tau)$ whenever $\tau \leq \tau^{\star}$ and $s(\tau) = 0$ when $\tau > \tau^{\star}$, where {$\tau^{\star} > 1$} is the unique threshold for which $\int_0^{\tau^{\star}} \min(1,1/\tau) 
g(\min(1,1/\tau)) d\tau ~=~1$. This leads to
\begin{align*}
\int_0^{\tau^{\star}} \min(1,1/\tau) g(\min(1,1/\tau)) d\tau ~&= ~ \int_0^{1}  g(1) d\tau +\int_1^{\tau^{\star}} 
\frac{1}{\tau} \cdot g(1/\tau) d\tau \\
~&=~ (1 - e^{-1})+\int_1^{\tau^{\star}} (1-e^{-1/\tau}) d\tau ~=~ 1\,.
\end{align*}
By solving the above equation numerically, we obtain $\tau^{\star} = 1.696$ and the 
optimal solution to Problem \eqref{lp:spm} is given by
\[\int_0^{\infty} s(\tau) d\tau ~=~ \int_0^{1} d\tau  +\int_1^{\tau^{\star}} 
\frac{1}{\tau}  d\tau  ~= ~1+ \ln(\tau^{\star}) ~=~ 1.5283\,.\]
Hence, the SPP mechanism that selects the best of Myersonian and uniform pricing rules yields at least $1/1.5283 ~\approx ~ 0.6543$ of $\Opt$, which is the bound in Theorem~\ref{thm:pp1}.
$\blacksquare$
\end{proof}

\begin{lemma} \label{lem:prob}
Let $Z_{\tau}$ be the number of buyers with $v_i ~\geq ~ t_i' ~\ge ~ \tau$; 
that is,  $Z_{\tau}= \sum_{i=1}^n\ind (v_i 
~\geq~  t_i' ~\geq ~ {\tau})$.
Then,
 \begin{align} \P[Z_{\tau} = 0] ~&\le~ \Q_n(s^{\star}(\tau))~\le~ \lim_{n\rightarrow 
 \infty}\Q_n(s^{\star}(\tau)) ~=~ e^{-s^{\star}(\tau)} 
 \label{eq:z_0}
\\
 2 \P[Z_{\tau} = 0] + \P[Z_{\tau} = 1] ~&\le~\p_n(s^{\star}(\tau))~\le~ 
 \lim_{n\rightarrow 
  \infty}\p_n(s^{\star}(\tau)) ~=~ (2 + s^{\star}(\tau)) e^{-s^{\star}(\tau)}\,, 
  \label{eq:z_0_z_1}
 \end{align}
 where $\Q_n(y) = \left( 1 - \frac{y}{n} \right)^n $ and $\p_n(y) = 2 \left( 1 
 - \frac{y}{n} \right)^n + y
\left( 1 - \frac{y}{n} \right)^{n-1} $.
\end{lemma}

Proof of Lemma \ref{lem:prob} is given in Section \ref{sec:proof:lem:prob}.

\subsection{Improved n-Dependent Bounds} \label{sec:small_n}
In numerous marketplaces, including online advertising markets, the number of buyers is rather small---for example in the {\color{black}display-ads} setting, the auctions are usually thin (not very many bidders). Motivated by this, in this section, we obtain improved bounds for the single-unit SPP mechanisms when the number of buyers $n$ is small. The technique is similar to the one used for the universal bound in Theorem~\ref{thm:pp1}. {\color{black}(Recall that the bound presented in Theorem \ref{thm:pp1} holds for any 
number of buyers $n$.)} 
In Figure \ref{fig:bound}, we illustrate our improved bounds for the SPP mechanisms. The figure also shows our improved bounds for ESP auctions, which will be discussed in Section \ref{sec:esp:n}. We also depict the best-known bound for these mechanisms prior to this work---that is, $1-(1-1/n)^n$---which is given by \cite{chawla2010multi}. We observe that our n-dependent bounds for the SPP mechanisms improve the prior bound by up to $3\%$. Further, the revenue bounds increase as the number of buyers decreases. 
\begin{figure}[h!]
\begin{center}
\includegraphics[width=3.5in]{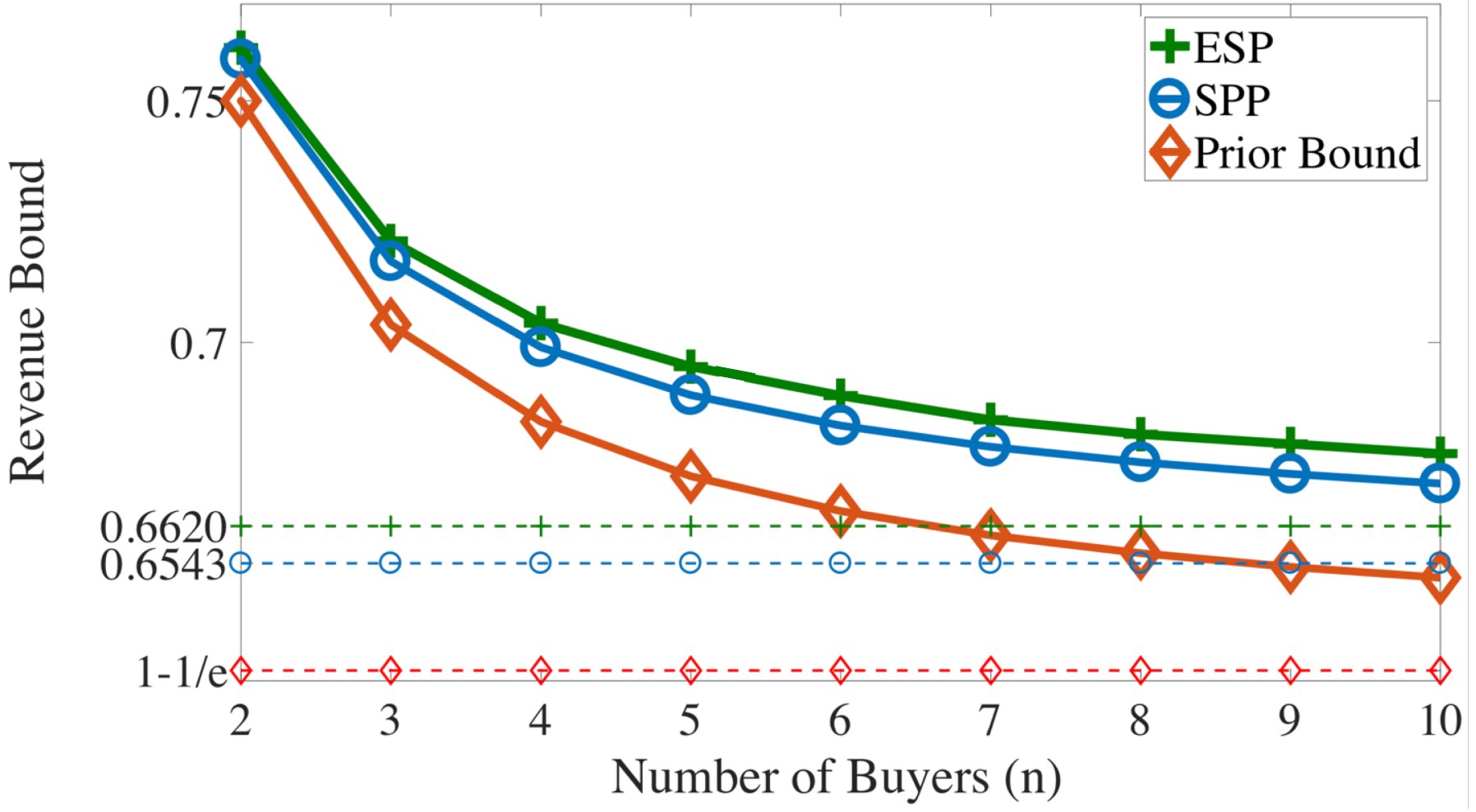} 
\caption{Comparing our bound with the best prior known bound for $n= 2, \ldots, 10$.  The red and green dashed curves represent the bounds in Theorems \ref{thm:pp1} and \ref{thm:pp2}, respectively. Recall that these bounds are valid for any $n\ge 1$. \label{fig:bound} }
\end{center}
\end{figure}

In this section, to highlight the  dependency of our bounds on the number of buyers $n$, we denote the revenue of SPP mechanisms with a vector of prices $\textbf{p} = (p_1, p_2, \ldots, p_n)$ by  $\PP_n(\textbf{p})$. We further denote the optimal revenue by $\Opt_n$.

The following Theorem~\ref{thm:pp_n} is the main result of this section. This theorem presents two approximation factors for the SPP mechanism with $n$ buyers: $\frac{1}{{{\text{\sf FR}}(n)}}$ and $\frac{1}{{{\text{\sf FR-d}}(n,k)}}$, where $k$ can be any positive integer value. See the statement of the theorem for the definition of ${{\text{\sf FR}}(n)}$ and ${{\text{\sf FR-d}}(n,k)}$. The first approximation factor ($\frac{1}{{{\text{\sf FR}}(n)}}$) is obtained by constructing a factor-revealing mathematical program (\ref{lp:spm_n_1}) that is very similar to that in Theorem \ref{thm:pp1}, thereby confirming the generalizability of our technique. However, evaluating $\frac{1}{{{\text{\sf FR}}(n)}}$ is not easy, as the mathematical program is non-linear.\footnote{The first approximation factor is provided  to help  readers see the unity across different settings.}
Further, in order to evaluate the approximation factor, we present a discretized version of this non-linear mathematical program, which is an easy-to-solve linear program (LP). This discretized LP yields an approximation factor of $\frac{1}{{\text{\sf FR-d}}(n,k)}$, where $k$ captures the granularity of discretization (the larger the $k$, the better the discretization). This quantity $\frac{1}{{\text{\sf FR-d}}(n,k)}$ can be easily computed for any $k$; see Table \ref{table:pp_finite_n}. (The ``d" in ${{\text{\sf FR-d}}(n,k)}$ stands for discretization.) We remark that the approximation factor of $\frac{1}{{\text{\sf FR-d}}(n,k)}$ is valid for any value of $k$; that is, Theorem \ref{thm:pp_n} presents a \emph{strong} factor-revealing mathematical program.\footnote{Factor-revealing mathematical programs that are not strong  present a valid bound only when $k$ goes to infinity.} Since larger values of $k$ imply more granular discretization, as seen in Table \ref{table:pp_finite_n}, our bound improves as $k$ increases.   

\begin{theorem}[$n$-Dependent Revenue Bound of SPP Mechanisms in Single-unit Settings] \label{thm:pp_n}
In a $1$-unit $n$-buyer setting with independent  private values, there {exists a vector of prices} $\mathtt{\bp} =(\mathtt{p}_1, \mathtt{p}_2, \ldots, \mathtt{p}_n)$, such that 
\begin{itemize}
\item 
	 \textbf{Non-discretized Bound.}   $\PP_n(\bp) \geq \Opt_n\cdot \frac{1}{{{\text{\sf FR}}(n)}}$, and 
	 \item \textbf{Discretized  Bound.} 
	 $\PP_n(\bp) \geq \Opt_n\cdot \frac{1}{{\text{\sf FR-d}}(n,k)}$ for any positive integer $k$,
\end{itemize}
 where $\Opt_n$ is the expected optimal revenue in a $1$-unit setting, $\PP_n(\bp)$ is the expected revenue of the SPP mechanism with prices $\mathtt{\bp}$, and $\text{\small{\sf{FR}}}(n)$ and ${{\text{\sf FR-d}}(n,k)}$ are defined as

\begin{multicols}{2}\setlength{\columnseprule}{1pt}
\noindent
  \begin{align*}
  &\text{\small{\sf{FR}}}(n) ~=~
  \\   
&  \begin{aligned}
\quad \max_{\{{s(\tau), \tau\ge 0\}}}& ~~\int_0^\infty 
s(\tau) d\tau\\
\text{s.t.} \quad &0~ \le s(\tau) \le \min(1, 1/\tau)~ & &\forall\tau \ge 0   \\[8pt]
& \int_0^\infty  \big(1-\Q_n(s(\tau))\big) d\tau ~\le~ 1 \\[8pt]
&\text{$s(\cdot)$ is weakly decreasing}\,.&
\end{aligned}
 \label{lp:spm_n_1} \tag{{{\small\sf{FR-n}}}}
\end{align*} 
\begin{align*}
&{\text{\small \sf FR-d}}(n,k)~=~ \\
  & \quad\begin{aligned}
 \max_{w} &
\sum_{i\in [k]} w_i & &\\
  \text{s.t.}\quad&{\sum_{i =j+1}^{k}} w_i \frac{ \mathtt{s}_j }{\mathtt{s}_i} ~\leq~ 1 & \forall j \in [k-1]& \\
  &\sum_{i=1}^k w_i \frac{1-\Q_n(\mathtt{s}_i)}{\mathtt{s}_i} ~\leq ~1&& \\
    & w_i ~\geq~ 0. &\forall i \in [k]&&
\label{lp:spm_n_dist} 
\end{aligned}
\tag{{\sf {\small{FR-n-d}}}}
  \end{align*}
\end{multicols}
Here, $\Q_n(y) = \left( 1 - \frac{y}{n} \right)^n $ and $\mathtt{s}_i = i/k$, $i\in [k]$. Furthermore, for $n\in [10]$ and $k\in\{200, 400, 800, 1600\}$, our approximation factors of $\frac{1}{{{\text{\sf FR-d}}(n,k)}}$ are presented in Table \ref{table:pp_finite_n}.
\end{theorem}
\begin{table}[h]
\setlength{\extrarowheight}{4pt}
\centering
\fontsize{9}{9}\selectfont{{
\begin{tabular}{|c|c|c||c|c||c|c||c|c||c|c|}
\hline
n & k & $\frac{1}{{\text{\sf FR-d}}(n,k)}$ &n & $\frac{1}{{\text{\sf FR-d}}(n,k)}$&n & $\frac{1}{{\text{\sf FR-d}}(n,k)}$ &n & $\frac{1}{{\text{\sf FR-d}}(n,k)}$ &n & $\frac{1}{{\text{\sf FR-d}}(n,k)}$  \\ \hline
\multirow{ 4}{*}{1} & 200 & {1.0000} & \multirow{ 4}{*}{2} & 0.7585 & \multirow{ 4}{*}{3} & 0.7167 &\multirow{ 4}{*}{4} &0.6988 &\multirow{ 4}{*}{5} & 0.6889  \\ 
& 400 & 1.0000 && 0.7586 &&0.7168 &&0.6989 && 0.6890\\
& 800 & {1.0000} && {0.7586} && {0.7168} && {0.6990} 
&&{0.6891}\\
&1600 & \textbf{1.0000} && \textbf{0.7586} && \textbf{0.7168} && \textbf{0.6990} && \textbf{0.6891}
\\ \hline  \hline   
\multirow{ 4}{*}{6} & 200 & 0.6826 &\multirow{ 4}{*}{7} & 0.6782 & \multirow{ 4}{*}{8} &0.6750 &\multirow{ 4}{*}{9} & 0.6726 &\multirow{ 4}{*}{10} & 0.6706  \\
& 400 & 0.6827 &&0.6784&&0.6752 &&0.6727 && 0.6708 \\
& 800 & {0.6828} && {0.6784} && {0.6752} &&{0.6728} && {0.6708}  \\ 
& 1600 & \textbf{0.6828} && \textbf{0.6785} && \textbf{0.6753} &&\textbf{0.6728} && \textbf{0.6709}\\
\hline
\end{tabular}
}}
\vspace{1em} 
\caption{$\frac{1}{{\text{\sf FR-d}}(n,k)}$ for $n \in [10]$ and $k = 200, 
400, 800$, and $1600$. For each $n\in [10]$, the maximum value of 
$\frac{1}{{\text{\sf FR-d}}(n,k)}$ (among the $k$'s considered) is boldfaced. 
\label{table:pp_finite_n}}
\end{table}

In the following account, we  present only  the proof of the first approximation factor---that is, the non-discretized bound of $\frac{1}{{{\text{\sf FR}}(n)}}$---and an overview on how the discretized LP \eqref{lp:spm_n_dist} is constructed. We provide the proof of the second approximation factor in Section \ref{sec:proof:thm:pp_n-dis}. 
\medskip 

\begin{proof}
{Proof of the Non-discretized Bound in Theorem \ref{thm:pp_n}.} 
 We use the same ideas as in the proof of Theorem~\ref{thm:pp1}. As the only difference, here, we bound $\MP$ using the $n$-dependent bound of Lemma~\ref{lem:prob}, namely, $\Q_n(\tau)$, rather than the $n$-independent limiting versions of these quantities (see Equation~\eqref{eq:z_0}) used in Theorem~\ref{thm:pp1}. In other words, in Equation \eqref{eq:u_lower_bound} in the proof of Theorem~\ref{thm:pp1}, instead of bounding $m(\tau)$ with $~1-\exp{(-s^\star({\tau}))}$, we bound it with $1- \Q_n(s^\star (\tau))$. Note that by Lemma~\ref{lem:prob}, $\exp{(-s^\star({\tau}))}$ and $\Q_n(s^\star (\tau))$ are the $n$-independent and $n$-dependent upper bounds, respectively, on the probability that there is no buyer $i$ with $v_i ~\geq ~ t_i' ~ \geq ~ \tau $. This leads to the second constraint in Problem \eqref{lp:spm_n_1}, as desired. The remainder of the proof is the same as that of Theorem~\ref{thm:pp1}; thus, it is omitted.
$\blacksquare$  \end{proof}

We now discuss how  Program \eqref{lp:spm_n_dist} is related to  Program \eqref{lp:spm_n_1}. 
 Recall 
that by Equation (\ref{eq:opt}), the optimal revenue is $\int_{0}^{\infty} 
s^\star(\tau) d\tau$, where $s^\star(\tau)$ is the probability that the optimal mechanism 
sells at a
price of at least $\tau$. We define 
$0 = \uptau_k \leq \uptau_{k-1} \leq \ldots \leq \uptau_1 \leq
\uptau_0 = \infty$ such that $\uptau_j =\inf\{\tau: s^\star(\tau) \le j/k\}$, $j\in [k-1]$. That is, we partition the range of $\tau =[0, \infty]$ using $s^\star(\cdot)$. Thus, as stated earlier, $k$ determines 
the granularity of our 
discretization. 
We then write 
\begin{align}\Opt ~=~ \sum_{i\in [k]} w^\star_{i}\,, \quad \quad \text{where} \quad w_{i}^\star ~=~ \int_{\uptau_{i}}^{\uptau_{i-1}} s^\star(\tau) d\tau\,. \label{eq:w}\end{align}
As we show in the proof of Theorem \ref{thm:pp_n}, $w_{i}^\star$, $i\in [k]$, is a feasible solution to Program \eqref{lp:spm_n_dist}. Consequently, the objective function of Program \eqref{lp:spm_n_dist} provides an upper bound on the optimal revenue. The first and second set of constraints in Problem \eqref{lp:spm_n_dist} can be viewed as the discretized version of the first and second set of constraints in Problem \eqref{lp:spm_n_1}, respectively. 
This discretization that we develop can likely be used in many other settings. 

\section{{Multi-unit SPP Mechanisms}}
\label{sec:multi-unit}
{In this section, we show that our pricing rules as well as our proof techniques generalize to the multi-unit settings. The principles underlying our pricing rules and the structure of the overall proof remain the same, while the actual proof itself is rather involved, requiring a few neat combinatorial lemmas.

In the multi-unit setting, there are $\hu \ge 1$ identical units of an item and $n$ unit-demand buyers. Similar to the previous section, the value of buyer $i\in [n]$ for being allocated is independently drawn from distribution $F_i$, where distributions are public knowledge. Note that under the SPP mechanism with price $\bp$, we approach buyers in decreasing order of prices and offer them a take-it-or-leave-it offer. Then, we continue selling the items until either the supply runs out or all buyers have been approached.  

{The following is the main result of this section.}

\begin{theorem}[Revenue Bound of SPP Mechanisms in Multi-unit Settings]\label{thm:pp_multi}
	In an $\hu$-unit $n$-buyer setting with independent private values, there {exists a vector of prices} $\mathtt{\bp} = (\mathtt{p}_1, \mathtt{p}_2, \ldots, \mathtt{p}_n)$, such that $
	 \PP_{\hu}(\bp) \geq \Opt_{\hu}\cdot \frac{1}{\text{\ref{lp:spmk}}}$, where $\Opt_{\hu}$ is the expected optimal revenue in an $\hu$-unit setting, $\PP_{\hu}(\bp)$ is the expected revenue of the SPP mechanism with prices $\mathtt{\bp}$, and {\sf{\small {FR-Multi}}({\hu})} is: 
  \begin{align} 
{\small{\text{\sf{\small FR-Multi}}({\hu})}} ~=~&{\max_{\{{s(\tau), \tau\ge 0\}}}} ~~\int_0^\infty 
s(\tau) d\tau \nonumber  \\
&  \begin{aligned}
\text{s.t.} ~~~ &0~ \le~ s(\tau) ~\le~ \min({\hu}, 1/\tau)~~ & &\forall ~~ \tau \ge 0 \nonumber \\
& \int_0^\infty  f_{\hu}(s(\tau)) d\tau ~\le~ 1\\
&\text{$s(\cdot)$ is weakly decreasing}\,.~~&\nonumber
\end{aligned}
\label{lp:spmk} 
 \tag{{{\small{\sf{FR-Multi({\hu})}}}}}
\end{align}
Here,  ${f_{\hu}(x) =  {\hu}- e^{-x} \sum_{i=0}^{{\hu}-1} ({\hu}-i)\frac{x^i}{i!} }$.  Our bound is greater than the best-known bound prior to this work---that is, $\frac{1}{{\text{\ref{lp:spmk}}}} > 1- \frac{{\hu}^{\hu}}{{\hu}! e^{\hu}}
$. 
\end{theorem}

 For $\hu\in [10], $ our approximation factor of $\frac{1}{{\text{\ref{lp:spmk}}}}$ is presented in Table \ref{table:multi} in Section~\ref{sec:intro}. The previous bound of $1- \frac{\hu^\hu}{\hu! e^{\hu}}$ was obtained via the correlation gap~\citep{yan2011mechanism} and independently by~\cite{TEGMM10} without using the correlation gap. Observe that when $\hu= 1$, the bound in Theorem \ref{thm:pp_multi} is the same as that in Theorem \ref{thm:pp1}. The proof of Theorem \ref{thm:pp_multi} is presented in Section \ref{sec:proof_multi}. We evaluate the performance of the 
$\PP$ mechanism that selects the best of Myersonian and uniform prices. 
{\color{black}Let $s^\star(\tau) = \sum_{i\in [n]} s_{i}^\star(\tau)$ be the expected number of buyers who are allocated the item  in the optimal mechanism and pay at least $\tau$, where $s_i^{\star}(\tau ) = \P[v_i~ \geq ~t_i~\geq ~{\tau}]$ and $t_i= t_i(\vminus)$'s are the thresholds in the optimal mechanism in the $\hu$-unit setting. 
To obtain the result, we show that $s^\star(\cdot)$ is a feasible solution to Problem {\text{\eqref{lp:spmk}}}.} Precisely, we first derive a lower bound on the revenue of the SPP mechanism with a uniform price, where this lower bound leads to the first set of the constraints in Problem {\text{\eqref{lp:spmk}}}. Then, we establish a lower bound on the revenue of the SPP mechanism with Myersonian prices as a function of $s^\star(\tau)$. This lower bound leads to the second constraint in Problem {\text{\eqref{lp:spmk}}}. 
Establishing this lower bound, which is presented in Lemma~\ref{lm:equal}, is one of the most challenging aspects of the proof. 

\begin{lemma}[Lower Bound of Myersonian SPP Mechanisms in Multi-unit Settings]\label{lm:equal} Consider the $\hu$-unit setting. 
Let ${s^\star(\tau )} = \sum_{i =1}^n \P[ v_i~ \geq ~t_i ~\geq ~{\tau}]$.  
Then, $m(\tau)$, which is the expected number of units that the Myersonian SPP mechanism sells with a price of at least ${\tau}$, satisfies the following inequality.
$$m(\tau) \geq \hu- \sum_{i=0}^{\hu-1} (\hu-i)\left( \begin{array}{c}
n  \\
i \end{array} \right)\frac{s^\star(\tau)^i}{n^i}(1-\frac{s^\star(\tau)}{n})^{n-i}.$$
\end{lemma}

\medskip 
We present the proof of Lemma \ref{lm:equal} in Section \ref{sec:equal}. In this lemma, we express $m(\tau)$ in the form of $n$-degree polynomials with $O(n^{\hu})$ terms and $n$ variables; by carefully grouping the terms in the polynomial, we show that the polynomial is minimized when all of its variables are equal. This yields the desired inequality. } 

{Theorem \ref{thm:pp_multi} also establishes that our bound $\frac{1}{{\text{\ref{lp:spmk}}}}$ is strictly better than the best-known bound prior to this work---that is, $1- \frac{{\hu}^{\hu}}{{\hu}! e^{\hu}}$. To show this result, in Lemma \ref{lem:LP_multi}, we first characterize {\text{\ref{lp:spmk}}}: 

\begin{lemma}[{\color{black}Characterization} of {\text{\ref{lp:spmk}}}] \label{lem:LP_multi} Consider any positive integer $\hu>1$.
 Let $\tau^{\star}> \frac{1}{{\hu}}$ be the unique solution of the following equation
 \begin{align*}
\int_{1/{\hu}}^{\tau^{\star}} \left({\hu}- e^{-1/\tau} \sum_{i=0}^{{\hu}-1} \frac{({\hu}-i)}{\tau ^i i!}\right)  d\tau ~=~ \frac{{\hu}^{\hu}}{{\hu}! e^{\hu}}\,.\end{align*}
Then, \ref{lp:spmk}, defined in Theorem \ref{thm:pp_multi}, {is given by} $1 + \ln({\hu}\tau^{\star})$.
 \end{lemma}
 
 The proof of Lemma \ref{lem:LP_multi} is presented in Section \ref{sec:proof:lem:LP_multi}. {\color{black}In that proof, we show that for any positive integer $\hu$, $g_{\hu}(x) := f_{\hu}(x)/x= \frac{{\hu}- e^{-x} \sum_{i=0}^{{\hu}-1} ({\hu}-i)\frac{x^i}{i!} }{x}$ is decreasing in $x$. We note that showing this monotonicity result is a non-trivial task because of the combinatorial term $\sum_{i=0}^{{\hu}-1} ({\hu}-i)\frac{x^i}{i!}$ in $f_{\hu}(x)$. Then, we show that $\ln({\hu} \tau^*)<\frac{{\hu}^{\hu}}{{\hu}! e^{\hu}}$. 
This inequality and Lemma \ref{lem:LP_multi} ensure that $\frac{1}{{\text{\ref{lp:spmk}}}} = \frac{1}{1+\ln({\hu} \tau^*)} > 1- \frac{{\hu}^{\hu}}{{\hu}! e^{\hu}} $, thereby confirming that our bound beats the best-known bound prior to this work.}

\section{{Position Auction, Partition Matroid, and General Matroid  SPP Mechanisms}}\label{sec:extend}
{The improved bounds of the multi-unit setting lead to improved approximations for position auctions and partition matroids settings, which we describe below. For general matroids, while we do not obtain improvements over the known approximation factor of $1-1/e$, we provide an alternate SPP mechanism using our pricing rules and obtain the same $1-1/e$ approximation factor.}
\subsection{{Position Auctions}}

The position auction (PA) setting captures the allocation constraints in search advertisements; see, for example, \cite{varian2007position} and \cite{athey2011position}. In this setting, there are $n$
positions characterized by click-through-rates $1\ge \alpha_1 \geq \alpha_2 \geq 
\hdots \geq \alpha_n \ge 0$ and $n$ buyers with private value-per-click equal to $v_1, \hdots,
v_n$, where $v_i$ is independently drawn from distribution $F_i$. If buyer $i$ is allocated to position $j$ and charged a certain payment amount $\pi_i$ in return for clicks, then his expected utility will be
$u_i = \alpha_j v_i - \pi_i$. Note that we describe the mechanism in terms of expected payment $\pi_i$ and not payment per click. If a buyer $i$ is allocated $x_i$ clicks in expectation and is charged
$\pi_i$, this is equivalent to charging him $\pi_i / x_i$ per click.

{An auction for the PA setting elicits bids from buyers and returns a (possibly randomized) allocation from buyers to positions and an (expected) payment for each buyer. Let $\mathcal A$ be the set of all feasible allocations from buyers to positions. Allocation $a\in \mathcal A$ is feasible if no more than one buyer is assigned to a position and no buyer is assigned to more than one position. Further, allocation $a$ can be represented by $J_i(a)$, $i\in [n]$, where $J_i(a)$ is the position that buyer $i$ is assigned under allocation $a$. Any direct mechanism $\mech$ can be described by its allocation and payment rules which we denote by $(\bf{y}, \boldsymbol{\pi})$, where ${\bf{y}}: \R^n\rightarrow [0,1]^{|\mathcal A|}$ and ${\boldsymbol{\pi}}: \R^n\rightarrow \R^n$. Here, $y_a(\hat {\bf{v}})$ is the probability that allocation $a$ is selected and $\pi_i(\hat{\bf{v}})$ is buyer $i$'s payment, given buyers' report $\hat {\bf{v}}$. Then, a direct mechanism is truthful if buyer $i$ prefers to reveal his true value; that is, for $v_i, \hat v_i$, we have 
\[\E\left[\sum_{a\in \mathcal A}y_a(v_i, {\bf{v}}_{-i}) \alpha_{J_i(a)} v_i -\pi_i(v_i, {\bf{v}_{-i}})\right] \ge \E\left [\sum_{a\in \mathcal A}y_a( \hat v_i, {\bf{v}}_{-i}) \alpha_{J_i(a)} v_i -\pi_i(\hat v_i, {\bf{v}}_{-i})\right]\,, \]
where the expectation is w.r.t. value of all the buyers except for buyer $i$---that is, ${\bf{v}}_{-i}$.
We let $x_i(\hat v_i, {{\bf{v}}_{-i}}) = \sum_{a\in \mathcal A}y_a(\hat v_i, {\bf{v}}_{-i}) \alpha_{J_i(a)}$ as the expected number of clicks of buyer $i$ when he reports $\hat v_i$ and other buyers are truthful.
Then, the mechanism is truthful if, for any $v_i, \hat v_i$, we have $\E[x_i(v_i, {\bf{v}}_{-i} )v_i -\pi(v_i, {\bf{v}_{-i}})]\ge \E[x_i(\hat v_i, {\bf{v}}_{-i})v_i -\pi(\hat v_i, {\bf{v}}_{-i})]$.}

{Observe that the truthfulness condition is expressed as a function of the expected number of clicks and payments. Given this, one may want to describe a mechanism in the PA settings as a mapping between a vector of values ${\bf v} =(v_1, \hdots, v_n)$ and expected clicks ${\bf x}(\cdot)$ and payments $\boldsymbol{\pi} (\cdot)$. However, there is a caveat: For a given vector of expected clicks $\bf x$, there may not always exist a randomized allocation. To make it clear, consider the following example. Assume that $n=2$ and let $\alpha_1 = 1$ and $\alpha_2 = 0.5$. Then, there is no randomized allocation that results in the expected clicks of ${\bf x}= (0.8, 0.8)$. This is so because $\sum_{i\in [2]}x_i$ exceeds the total available click-through-rates---that is, $\sum_{i\in [2]}x_i> \sum_{i\in [2]}\alpha_i = 1.5$.
Lemma 1 in \cite{ feldman2008truthful} nicely generalizes this observation and provides necessary and sufficient conditions under which a vector of expected clicks is feasible.}  

{\begin{lemma}[Feasibility of Expected Clicks -- Lemma 1 in \cite{ feldman2008truthful}]\label{lemma:pa_feasibility} In a PA setting with click-through-rates $\alpha_1 \geq \alpha_2 \geq \hdots \geq \alpha_n$, a vector of expected clicks ${\bf x} =(x_1, \ldots, x_n)$ is feasible---that is, there exists a randomized allocation $y\in [0,1]^{|\mathcal A|}$ such that $x_i=\sum_{a\in \mathcal A}y_{a}\alpha_{J_i(a)}$ for any $i\in [n]$, if and only if the following inequalities are satisfied:
  \begin{equation}\label{eq:pa_feasiblity}
  \sum_{i \in S} x_i \leq \sum_{j=1}^{\abs{S}} \alpha_j, \qquad \forall S
  \subseteq [n]\,.
  \end{equation}
\end{lemma}}

{As indicated in \cite{feldman2008truthful}, constructing randomized allocation $y$ for any valid vector of expected clicks is closely related to the classical scheduling problems, which are studied in \cite{horvath1977level} and \cite{gonzalez1978preemptive}. For more details, see \cite{feldman2008truthful}.}

{Lemma \ref{lemma:pa_feasibility} enables us to describe a (truthful) PA mechanism as a mapping between a vector of values ${\bf v} =(v_1, \hdots, v_n)$ and expected clicks ${\bf x}$ and payments $\boldsymbol{\pi}$, such that for any $\bf v$, the expected clicks satisfy Inequalities (\ref{eq:pa_feasiblity}). With this in mind, in the following account, we first present the optimal mechanism for the PA settings. In particular, we show that the optimal mechanism can be expressed as a function of the allocation and payment rules of the optimal mechanism in the $j$-unit setting with the same value distributions, where $j \in [n]$. 
Inspired by this, we then present an SPP mechanism for the PA settings that builds on our proposed SPP mechanisms for the $j$-unit setting (see Section \ref{sec:multi-unit}). As our main result, we show that our SPP mechanism for the PA settings obtains an approximation of $0.6543$ to the optimal revenue. While $0.6543$ is a universal bound for any click-through-rates, we obtain an improved bound of $\sum_{j=1}^n f_j /
\text{\sf{FR-Multi}}(j)$, where $f_j$ is the fraction of the optimal revenue that can be attributed to position $j$ (precisely defined later in this section) and $1/\text{\sf{FR-Multi}}(j)$ is the approximation ratio proved for our SPP mechanisms in the $j$-unit settings.}

{Next, we characterize the optimal mechanism. For $j\in [n]$ and $i\in [n]$, let $x_i^j({\bf{v}}) \in \{0,1\}$ and $\pi_i^j(\bf{v}) \in \R_+$ be the allocation and payments in the optimal $j$-unit mechanism when buyers' value is ${\bf v}=(v_1, \ldots, v_n)$. It is easy to see that $x_i^1({\bf{v}}) \leq x_i^2({\bf{v}}) \leq \hdots \le x_i^n({\bf{v}})$ for every buyer $i$. Consequently, there is a position $J_i$, such that $x_i^{j}({\bf{v}}) = 1$ for any $j \geq J_i$, and zero otherwise. This is so because in the $j$-unit optimal mechanism, items are allocated to at most $j$ bidders with the highest non-negative (ironed) virtual values. Thus, if a buyer $i$ is allocated in the $j$-unit optimal mechanism, he is also allocated in the $j'$-unit optimal mechanism, where $j'> j$. Now, consider the following auction in the PA settings that assigns position $J_i$ to buyer $i$; that is, the buyer gets the expected click of $x_i({\bf{v}}) = \sum_{j\in [n]} (\alpha_j - \alpha_{j+1}) x_i^j({\bf{v}}) = \alpha_{J_i}$ and is charged $\pi_i({\bf{v}}) = \sum_{j\in [n]} (\alpha_j - \alpha_{j+1}) \pi_i^j(\bf{v})$. The next lemma reveals that this auction is indeed optimal in the PA settings.}  

{\begin{lemma} [Optimal Mechanism in Position Auction Settings]\label{lemma:pa_structure}
  For $j\in [n]$ and $i\in [n]$, let $x_i^j({\bf{v}}) \in \{0,1\}$ and $\pi_i^j(\bf{v}) \in \R_+$ be the
 allocations and payments in the $j$-unit optimal mechanism, when buyers' value is ${\bf v}=(v_1, \ldots, v_n)$. Then, the mechanism for the PA
 settings with the following rules is optimal:
 $$x_i({\bf{v}}) = \sum_{j\in [n]} (\alpha_j - \alpha_{j+1}) x_i^j({\bf{v}}) \quad \text{and} \quad \pi_i({\bf{v}}) = \sum_{j\in [n]} (\alpha_j - \alpha_{j+1}) \pi_i^j(\bf{v})\,,$$  
  where $\alpha_{n+1} = 0$.  
\end{lemma}}

{The proof of Lemma \ref{lemma:pa_structure} is presented in Section \ref{sec:proof:lem:pos}. This lemma enables us to view an auction for the PA setting as a combination of the multi-unit auctions. It provides the following decomposition of the optimal revenue:
$$\textstyle \E[\sum_{i\in [n]} \pi_i ({\bf{v}})] = \sum_{j\in [n]} (\alpha_j - \alpha_{j+1}) \cdot
\E[\sum_{i\in [n]} \pi_i^j({\bf{v}})]\,,$$
where the expectation is w.r.t. buyers' value. Note that the $j$-th term is the contribution of the $j$-unit auction---that is, the $j$-th position, to the total revenue. Precisely, we 
 define
\begin{align}\label{eq:f}f_j = \frac{(\alpha_j - \alpha_{j+1}) \cdot \E[\sum_{i\in [n]} \pi_i^j({\bf{v}})]}{\E[\sum_{i\in [n]}
\pi_i({\bf{v}})] }\end{align} as the fraction of the optimal revenue that can be attributed to position $j\in [n]$.\medskip

\textbf{SPP Mechanisms for Position Auction Settings.}
Motivated by the structure of the optimal mechanism in Lemma \ref{lemma:pa_structure}, we propose the following SPP mechanism: For each $j=1, \ldots, n$, we run the SPP mechanism with the best of Myersonian and uniform prices for $j$-unit settings, as described in Section \ref{sec:multi-unit}. Let $\tilde{x}_{i}^j({\bf{v}})$ and $\tilde{\pi}_{i}^j({\bf{v}})$, $i\in [n]$, be the outcome of the SPP mechanism in the $j$-unit settings.} {That is, $\tilde{x}_{i}^j({\bf{v}})$ is one if buyer $i$ gets an item in the SPP mechanism for $j$-unit setting, and zero otherwise. Further, $\tilde{\pi}_{i}^j({\bf{v}})$ is buyer $i$'s payment in that auction. 
We then define 
  \begin{align}\tilde{x}_i({\bf{v}}) = \sum_{j\in [n]} (\alpha_j - \alpha_{j+1}) \tilde{x}_i^j({\bf{v}}) \quad \text{and} \quad 
  \tilde{\pi}_i({\bf{v}}) = \sum_{j\in [n]} (\alpha_j - \alpha_{j+1}) \tilde{\pi}_i^j({\bf{v}})\label{eq:spm_PA}\end{align} as the expected number of clicks and the expected payment of buyer $i$ in the SPP mechanism for the position auction when buyers' value is $\bf{v}$. 
} 

{At first glance, it may not be obvious that the described mechanism is an SPP mechanism for the PA settings. However, in fact, the mechanism can be explained as an SPP mechanism with $n^2$ prices. Let $p_i^j$ be the best of Myersonian and uniform prices for buyer $i$ in the SPP mechanism with $j$ units. Then, for each position $j =1, 2, \ldots, n$, we approach buyers sequentially in decreasing order of their prices $p_i^j$ and offer them the expected number of clicks of $(\alpha_j-\alpha_{j+1})$ at price of $p_i^j(\alpha_j-\alpha_{j+1})$. We stop when either $j$ buyers accept our offer or we have approached them all.}

{One can think about $(\alpha_j-\alpha_{j+1})$ as the extra clicks that a buyer obtains when he is moved from position $j+1$ to position $j$. That being said, when a buyer accepts the offer associated with position $j$, this does not imply that he will be assigned to position $j$. Put differently, when we approach buyers, we do not offer them a particular position; instead, we offer them an (extra) expected number of clicks. This enables us to run the SPP mechanisms in parallel or sequentially in an arbitrary order, as the SPP mechanisms for different positions are independent of each other. Because of this, after we run all the SPP mechanisms, we may have a buyer $i$ who has accepted two offers, one for position $1$ and one for position $3$. This implies that in this SPP mechanism for the PA setting, the buyer obtains expected clicks of $(\alpha_1-\alpha_2)+(\alpha_3-\alpha_4)$ at the expected price of $p_i^1(\alpha_1-\alpha_2)+p_i^3(\alpha_3-\alpha_4)$.  }

{Thus, it is evident that the mechanism is truthful, in the sense that when each buyer $i$ is approached for the $j$-unit auction at price $p_i^j$, he accepts the offer when his value-per-click $v_i$ is greater than or equal to $p_i^j$. This is so because (i) when $v_i\ge p_i^j$, the extra utility that the buyer enjoys from accepting the offer---that is, $(\alpha_j-\alpha_{j+1})(v_i-p_i^j)$---is non-negative and vice versa, and (ii) accepting or rejecting an offer does not influence other offers.  }

{Next, we show that  $(\tilde{x}_1({\bf{v}}), \hdots, \tilde{x}_n({\bf{v}}))$, defined in Equation (\ref{eq:spm_PA}) is a valid vector of expected clicks, in the sense that it satisfies the feasibility conditions in
 (\ref{eq:pa_feasiblity}). Recall that the  feasibility conditions in (\ref{eq:pa_feasiblity}) are necessary and sufficient to have  a randomized allocation $y({\bf{v}})\in [0,1]^{|\mathcal A|} $ such that for any $i\in [n]$, we have $\tilde x_i({\bf{v}})=\sum_{a\in \mathcal A}y_{a}({\bf{v}})\alpha_{J_i(a)}$.   }

{\begin{lemma}[Feasibility of Expected Clicks in the SPP Mechanism]\label{lem:feasible}
  Suppose that for any $i, j\in [n]$,  $x_i^j \in [0,1]$ and for any $j \in [n]$, we have $\sum_{i\in [n]} x_i^j \leq j$. Then, $x_i = \sum_{j\in [n]}
  (\alpha_j - \alpha_{j+1}) x_i^j$, satisfies the feasibility conditions in (\ref{eq:pa_feasiblity}). 
\end{lemma}

\begin{proof}{Proof of Lemma \ref{lem:feasible}}
  For any subset $S \subseteq [n]$, we have $$\textstyle\sum_{i \in S} x_i = \sum_{j\in[n]}
  \left( (\alpha_j - \alpha_{j+1}) \sum_{i \in S} x_i^j \right) \leq \sum_{j\in [n]} (\alpha_j -
  \alpha_{j+1}) \min(j, \abs{S}) = \sum_{j=1}^{\abs{S}} \alpha_j\,,$$
  where the inequality holds because $\sum_{i\in [n]} x_i^j \leq j$. 
  The above equation verifies Condition  (\ref{eq:pa_feasiblity}) and completes the proof. 
$\blacksquare$ \end{proof}

We now present the approximation factor for our SPP mechanism. 

\begin{theorem}[Revenue Bound of SPP Mechanisms in Position Auction Settings]\label{thm:PA}
  Our SPP mechanism for PA Settings defined above is a $\sum_{j\in [n]} \frac{f_j} 
  {\text{{\sf{FR-Multi}}}(j)} \geq \frac{1} {\text{{\sf{FR-Multi}}}(1)}  = 0.6543$-approximation,
  where $f_j$'s, defined in Equation (\ref{eq:f}), are the fraction of the optimal revenue attributed to position $j$, and $\text{\sf{\small{FR-Multi}}}(j)$ is defined in Theorem \ref{thm:pp_multi}.
\end{theorem}

\begin{proof}{Proof of Theorem \ref{thm:PA}}
  Let $(x_i(\cdot),\pi_i(\cdot))$, $i\in [n]$, represent the  expected number of clicks and payment in the optimal PA mechanism and let $(\tilde{x}_i(\cdot),\tilde{\pi}_i(\cdot))$, $i\in [n]$, be
  the expected number of clicks and payment in our SPP mechanism for the PA settings. Finally, for $j, i\in[n]$, let $(x_i^j(\cdot),\pi_i^j(\cdot))$, and $(\tilde x_i^j(\cdot),\tilde \pi_i^j(\cdot))$  be  their respective decompositions in terms of multi-unit auctions (see Lemma \ref{lemma:pa_structure} and Equation (\ref{eq:spm_PA})). Then, we have
  \begin{align*}\E\Big[\sum_{i\in [n]} \tilde{\pi}_i ({\bf v})\Big] &= \sum_{j\in [n]} (\alpha_j - \alpha_{j+1})
  \E\left[\tilde{\pi}_i^j({\bf v})\right] \\
  &\geq \sum_{j\in [n]} (\alpha_j - \alpha_{j+1}) \frac{
    \sum_{i\in [n]} \E[\pi_i^j({\bf v})]}{ \text{\small{\sf{FR-Multi}}}(j) } =\sum_{j\in [n]} \frac{f_j}{\text{\small{\sf{FR-Multi}}}(j) }  
  \E\Big[\sum_{i\in [n]}
\pi_i({\bf{v}})\Big]\,, \end{align*}
where the inequality follows from Theorem \ref{thm:pp_multi}, which provides an approximation factor for the SPP mechanism in the $j$-unit setting, and the last equation holds because of the definition of $f_j$, provided in Equation (\ref{eq:f}). The bound of $\frac{1} {\text{{\sf{FR-Multi}}}(1)}$ can be  obtained by observing that $\sum_{j\in [n]} \frac{f_j}{{\text{\small{\sf{FR-Multi}}}}(j) } \ge \frac{1} {\text{{\sf{FR-Multi}}}(1)}$. $\blacksquare$  
\end{proof}}

\subsection{{Partition Matroid Settings}}\label{sec:partitionmatroid}
The partition matroid feasibility constraint is defined by a partition of the set of buyers $[n]$ into $[n] = S_1 \cup S_2 \cup \dots \cup S_k$ that is publicly available, and from each set $S_i$, at most, $\hu_i$ buyers are allowed to be allocated. Such feasibility constraints could arise from, for example, legal/policy restrictions that prevent more than $\hu_i$ buyers from a particular geographical region $i$ from receiving an allocation.

The definition of SPP mechanism is provided in Section~\ref{sec:model} for general matroids. {Just as in the multi-unit setting, we choose between Myersonian posted prices (MP) and uniform posted price (UP) in order to obtain the better expected revenue-yielding mechanism when we select our prices---we do this on a per set basis here. In other words, for each set $S_i$, we choose between MP and UP based on revenue in order to select the prices for that set. It is immediately apparent that the SPP mechanism approximation factor for partition matroids strictly exceeds $\min_i (1 - \frac{\hu_i^{\hu_i}}{\hu_i!e^{\hu_i}})$ and matches the numbers in Table~\ref{table:multi} for the smallest $\hu_i$.}

\subsection{ {Matroid Settings}}\label{sec:matroid}
{From the techniques of~\cite{yan2011mechanism}, it follows that our Myersonian SPP mechanism already obtains a $1-1/e$ approximation to $\OptM$. \citeauthor{yan2011mechanism} establishes that the expected revenue of the optimal mechanism (Myerson's mechanism) is upper bounded by $\E_W[f(W)]$, where $W$ is the set of winners in the optimal mechanism and $f(\cdot)$ is the weighted matroid rank function, which happens to be a monotone submodular function. ~\citeauthor{yan2011mechanism} also shows that the revenue of any SPP mechanism (including our Myersonian SPP mechanism) is $\E_S[f(S)]$, where $S$ is the set of buyers that exceed their posted prices and $f(\cdot)$ is the same weighted matroid rank function. The commonality between $S$ and $W$ is that if a buyer $i$ is an element of $W$, with probability $q_i$, the buyer $i$ is an element of $S$ with the same probability $q_i$. This follows from how the thresholds in the Myersonian  prices are constructed from Myerson's mechanism itself, using the taxation principle. The difference between $S$ and $W$ is that the elements of $S$ are independently selected (recall that we re-sample other buyers' values when selecting the threshold for each buyer), whereas the elements of $W$ are correlated. A beautiful result regarding submodular functions~\citep{Vondrak07, ADSY10} states that the correlation gap of submodular functions is, at most, $\frac{e}{e-1}$. In other words, $ \frac{\E_{S\sim D}[f(S)]}{\E_{S\sim I(D)}[f(S)]} \leq \frac{e}{e-1}$, where $D$ is an arbitrary joint distribution over the ground set of elements over which the submodular function $f(\cdot)$ is defined, and $I(D)$ is an independent distribution over the ground set with the same marginals as $D$. This indicates that our Myersonian SPP mechanism obtains at least a $1-\frac{1}{e}$ fraction of $\OptM$, thereby proving Proposition~\ref{prop:matroid}.
\begin{proposition}[Revenue Bound of SPP Mechanisms in General Matroid Settings]
	\label{prop:matroid}
In an $n$-buyer setting with independent private values and any matroidal feasibility constraints, the Myersonian SPP mechanism obtains a revenue of at least $(1-\frac{1}{e})\cdot\OptM$.
\end{proposition}}

\section{Eager Second-price Auctions}\label{sec:eager}
In this section, we show that our pricing rules also lead to improved approximation bounds for the $1$-unit ESP auctions. {\color{black}The central proof technique and principles are similar to what was used  earlier for our SPP mechanisms.} We note that here, we focus on  \emph{eager} second-price
auctions as opposed to {\color{black}\emph{lazy} second-price auctions}.  The lazy second-price auctions {neither dominate nor are dominated by} ESP auctions for general correlated distributions but are within a factor of $2$ of each other~\citep{paes2016field}. Further, \citet{paes2016field} show that the optimal revenue from the ESP auction dominates the optimal revenue from the {\color{black}{lazy} second-price auctions} when the 
value distributions are independent. Motivated by this, we  study ESP auctions in the current work. We note that it is known from an example in Section 4 of \cite{ronen2001approximating} that it is impossible to obtain an approximation that is better than $1/2$ for optimal revenue via lazy second-price auctions. We now proceed to formally define ESP auctions. 

\textbf{{Eager Second-Price Auctions $\EG(\bp)$}}
\begin{itemize}
\item [-] {Each buyer $i\in[n]$ submits his bid, which is equal to his value $v_i$.} 
\item [-] {All the buyers with value $v_i < p_i$ are  eliminated first. Let
$S =\{i: v_i ~{\ge} ~ p_i\}$ be the set of all the buyers who clear their reserve
prices.}
 \item[-] The item is then allocated to buyer $i^{\star} =
\arg\max_{i\in S} v_i$, who has the highest value among all buyers in set
$S$, and he pays $\max (p_{i^{\star}},\, \max_{i\in S, i\ne
i^{\star}}v_i)$. For other buyers, their payment is zero.  
\end{itemize}
\medskip

Note that ESP auctions are truthful in the dominant strategy sense. Thus, for each buyer $i$, his bid is equal to his value, regardless of the submitted bids of other buyers. 

Lemma  \ref{lemma:rev_pp_eager} shows that our bounds for the SPP mechanisms in the single-unit setting---presented in Theorems \ref{thm:pp1} and \ref{thm:pp_n}---are also valid bounds for ESP auctions. This lemma is an important observation regarding the revenue of  $\EG(\bp)$  and  $\PP(\bp)$ made by \cite{chawla2010multi}.

\begin{lemma}[ESP vs SPP -- Theorem 32 in arxiv Version v2 of~\cite{chawla2010multi}]\label{lemma:rev_pp_eager}
{In an $n$-buyer setting, for any vector of prices $\bp = (p_1, \hdots, p_n)$ and any value distributions, the revenue of $\EG(\bp)$ is at least the revenue of $\PP(\bp)$ in single-unit settings.}
\end{lemma}

 However, ESP auctions can earn higher revenue than SPP mechanisms by leveraging the second-highest bid. In this section, we show how to exploit the second-highest bid to obtain an improved 
bound for ESP auctions.

\subsection{A Universal Bound for Eager Second-price Auctions}
{\color{black}Theorem \ref{thm:pp2} presents a universal bound for eager second-price auctions.} Similar to Theorem \ref{thm:pp_n}, this theorem presents two approximation factors: $\frac{1}{\text{\sf{FR-ESP}}}$ and $\frac{1}{\text{\sf{FR-ESP-d}}(k)}$. See the statement of the theorem for the definition of $\text{\sf{FR-ESP}}$ and $\text{\sf{FR-ESP-d}}(k)$. We obtain the first approximation factor by establishing a factor-revealing mathematical program using the decision variable $s(\cdot)$. The second approximation factor is derived by discretizing the aforementioned mathematical program. As in the earlier theorems, the discretization here is solely for the purposes of evaluating the approximation factor.\medskip

\begin{theorem}[Revenue Bound of ESP Auctions in Single-unit Settings]\label{thm:pp2}
In a single-unit $n$-buyer setting with independent private values, there {exists a vector of prices} $\mathtt{\bp} =(\mathtt{p}_1, \mathtt{p}_2, \ldots, \mathtt{p}_n)$, such that \begin{itemize}\item \textbf{{Non-discretized Bound.}} $\EG(\mathtt{\bp}) ~\geq~
	 \Opt \cdot \frac{1}{{\text{\sf 
FR-ESP}}}$, and 
\item \textbf{{Discretized Bound.}} $\EG(\mathtt{\bp}) ~\geq~\Opt \cdot \frac{1}{{\text{\sf \small{FR-ESP-d}}}(k)}$ for any positive integer $k$,
\end{itemize}
where $\Opt$ is the expected optimal revenue, $\EG(\mathtt{\bp})$ is the expected revenue of an ESP auction with personalized reserve prices $\mathtt{\bp}$, and $\text{\sf 
FR-ESP}$ and ${{\text{\sf \small{FR-ESP-d}}}(k)}$ are defined as
 \begin{multicols}{2} \setlength{\columnseprule}{1pt}
\noindent
\begin{align*} 
&\text{\sf{FR-ESP}} ~=~\\ &{\max_{\{{s(\tau), \tau\ge 0\}}}} ~\int_0^\infty 
s(\tau) d\tau\quad \text{s.t.} \nonumber\\[8pt]
&  \begin{aligned}
 &\int_{\T_x}^\infty  (2-2e^{-s(\tau)} - s(\tau)e^{-s(\tau)}) d\tau \\[8pt]
 &~~+\int_0^{\T_x} ( x +  (1-e^{-s(\tau)})) d\tau \le2~~~~  \forall x\in [0,1] \nonumber \\
& \int_0^\infty  f(s(\tau)) d\tau ~\le~ 1\\[8pt]
&\text{$s(\cdot)$ is weakly decreasing}\,,~~&\nonumber
\end{aligned}
 \label{lp:esp} \tag{\small{\sf{FR-ESP}}}
\end{align*}
\begin{align*}
&{\text{\sf \small{FR-ESP-d}}}(k)= \\
& \max_w \sum_{i\in [k]} w_i \quad \text{s.t.} \nonumber \\
&\begin{aligned}   &\sum_{i=1}^j  w_i 
\frac{2(1-e^{-\mathtt{s}_i}) -
 \mathtt{s}_ie^{-\mathtt{s}_i}}{\mathtt{s}_i} \\ 
\quad &~~~+ \sum_{i=j+1}^{k} w_i \frac{ \mathtt{s}_j+(1-e^{-\mathtt{s}_i}) }{\mathtt{\mathtt{s}}_i}
    ~\leq~ 2,  &\forall j \in [k] \\
   & \sum_{i\in [k]} w_i \frac{1-e^{-\mathtt{s}_i}}{\mathtt{s}_i}~ \leq ~ 1 & \\
   & w_i \geq 0\,. &\forall i \in [k]
\end{aligned}
\label{lp:esp:d} \tag{\small{\sf FR-ESP-d}}
\end{align*}
\end{multicols}
Here, $f(x) = ({1 - e^{-x}})$, for any $x\in [0,1]$,  $\T_x~=~  
\inf\{\tau: s^\star(\tau) \le  x\}$, and  $\mathtt{s}_i = i/k$, for $i\in [k]$. 
Further, setting $k=3200$, the approximation factor is 
$\frac{1}{\text{\small{\sf FR-ESP-d}}(3200)} = 0.6620$. 
\end{theorem}

The proof of Theorem \ref{thm:pp2} is provided in Section \ref{sec:proof:claim}.
{The proof of the first approximation factor---that is, the non-discretized bound---is similar to that of Theorem \ref{thm:pp1}. We consider an ESP auction that chooses the best of 	{Myersonian and uniform reserve prices.} (The definition of Myersonian and uniform ESP auctions is presented in Section \ref{sec:proof:claim}.) By constructing a factor-revealing mathematical program, we show that the ratio of the  maximum revenue from the Myersonian ESP auction and uniform ESP auction to the optimal revenue $\Opt$ is at least $\frac{1}{\text{\sf FR-ESP}}$.}  Observe that Problem \eqref{lp:esp} is similar to Problem \eqref{lp:spm}. The main difference between these two problems is their first set of constraints: the first set of constraints in Problem \eqref{lp:esp} is obtained by lower-bounding the sum of the revenue of the ESP with the Myersonian pricing rule and the revenue of the ESP with a uniform price of $\T_x$, where $\T_x~=~  
\inf\{\tau: s^\star(\tau) \le  x\}$ and $x\in [0,1]$. (Recall that the first set of constraints in Problem \eqref{lp:spm} is obtained by lower-bounding the revenue of the SPP mechanism with the uniform pricing rule.)

The solution to Problem \eqref{lp:esp} cannot be easily obtained, as it also depends on $\T_x~=~  
\inf\{\tau: s^\star(\tau) \le  x\}$. Thus, to solve this problem, we again use our discretization technique to convert it to an easy-to-solve LP. This leads to the our discretized bound $\frac{1}{{\text{\sf \small{FR-ESP-d}}}(k)}$. Table \ref{table:reg} presents the values of 
$\frac{1}{\text{\sf \small{FR-ESP-d}}(k)}$ for  different values of $k$. Since 
$\frac{1}{\text{\sf \small{FR-ESP-d}}(k)}$ is a valid approximation factor for every $k$, 
it follows that $\frac{1}{\text{\sf \small{FR-ESP-d}}(3200)} = 0.6620$ is a valid 
approximation factor. As earlier, parameter $k$ determines the precision of our discretization. 

\begin{table}[h]
\setlength{\extrarowheight}{4pt}
\centering
\fontsize{9}{9}\selectfont{{
\begin{tabular}{|c |c |c |c|c|c|c|c|} 
 \hline
 $k$ & 50 & 100 & 200 & 400 & 800 & 1600 &3200 \\ 
 \hline
 $\frac{1}{\text{\sf FR-ESP-d}(k)}$ & 0.6606 & 0.6613 & 0.6617 &0.6618 &  0.6619 &0.6620 & 0.6620\\ \hline
\end{tabular}
}}
\vspace{1em} 
\caption{$\frac{1}{\text{\sf FR-ESP-d}(k)}$ for different values of $k$. }
\label{table:reg}
\end{table}

\subsubsection{Improved n-Dependent Bounds for Eager Second-Price Auctions}\label{sec:esp:n}
ESP auctions are widely used in the online advertising market. In this market, as stated earlier, because of targeting and the heterogeneous preferences of buyers (advertisers), the number of buyers is rather small. Motivated by this, in the following theorem, we present improved $n$-dependent bounds for ESP auctions.  
The gap between $n$-dependent bounds and our universal bounds is bigger when $n$ is smaller; see Table \ref{tab:smalln}.
  
\begin{table}[h]
\setlength{\extrarowheight}{4pt}
\centering
\fontsize{9}{9}\selectfont{{
\begin{tabular}{|c|c|c||c|c||c|c||c|c||c|c|}
\hline
n & k & $\frac{1}{{\text{\sf FR-ESP-d}}(n,k)}$ &n & $\frac{1}{{\text{\sf FR-ESP-d}}(n,k)}$&n  & $\frac{1}{{\text{\sf FR-ESP-d}}(n,k)}$ &n & $\frac{1}{{\text{\sf FR-ESP-d}}(n,k)}$ &n  & $\frac{1}{{\text{\sf FR-ESP-d}}(n,k)}$   \\ \hline
\multirow{ 3}{*}{1} & 200 & {1.0000} & \multirow{ 3}{*}{2} & 0.7610 & \multirow{ 3}{*}{3} & 0.7207 &\multirow{ 3}{*}{4} &0.7038 &\multirow{ 3}{*}{5} & 0.6944   \\ 
& 400 & 1.0000 && 0.7611 &&0.7209  &&0.7039 && 0.6945\\
& 800 & {1.0000}  && {0.7611} &&  {0.7209 } && {0.7040} 
&&{0.6946}\\
& 1600 & \textbf{1.0000}  && \textbf{0.7611} &&  \textbf{0.7210 } && \textbf{0.7040} 
&&\textbf{0.6946}
\\ \hline   \hline      
\multirow{ 4}{*}{6} & 200 & 0.6884 &\multirow{ 4}{*}{7} & 0.6843  & \multirow{ 4}{*}{8}  &0.6813 &\multirow{ 4}{*}{9} &  0.6790  &\multirow{ 4}{*}{10} & 0.6771   \\
& 400 & 0.6886 &&0.6844&&0.6814 && 0.6791 && 0.6773 \\
& 800 & {0.6886} && {0.6845} && {0.6815} &&{ 0.6792} && {0.6774}  \\ 
& 1600 & \textbf{0.6887} && \textbf{0.6846} && \textbf{0.6815} &&\textbf{0.6792} && \textbf{0.6774}\\\hline
\end{tabular}
}}
\vspace{1em} 
\caption{$\frac{1}{{{\text{\sf FR-ESP-d}}(n,k)}}$ for $n \in [10]$ and  $k = 200, 400, 800$, and  1600.  \label{table:esp_finite_n}}
\end{table}

\begin{theorem}[$n$-Dependent Revenue Bound of ESP Auctions in Single-unit Settings] \label{thm:esp_n} 
In a single-unit $n$-buyer setting with independent private values, there {exists a vector of prices} $\mathtt{\bp} =(\mathtt{p}_1, \mathtt{p}_2, \ldots, \mathtt{p}_n)$, such that
 \begin{itemize}\item \textbf{{Non-discretized Bound.}} $
	\EG_n(\mathtt{\bp}) ~\geq~
	\Opt_n\cdot \frac{1}{{\text{\sf \small{FR-ESP}}}(n)}$, and 
\item \textbf{{Discretized Bound.}} $\EG_n(\mathtt{\bp}) ~\geq~\Opt_n\cdot \frac{1}{{\text{\sf \small{FR-ESP-d}}}(n,k)}$ for any positive integer $k$,
\end{itemize}
	 where $\Opt_n$ is the expected optimal revenue, $\EG_n(\mathtt{\bp})$ is the expected revenue of an ESP auction with personalized reserve prices  $\mathtt{\bp}$, and $\text{\sf{FR-ESP}}(n)$ and  ${\text{\sf \small{FR-ESP-d}}}(n,k)$ are defined as
	  \begin{multicols}{2} \setlength{\columnseprule}{1pt}
\noindent
\begin{align*} 
&\text{\sf{FR-ESP}}(n) ~=~\\[4pt] &{\max_{\{{s(\tau), \tau\ge 0\}}}} ~\int_0^\infty 
s(\tau) d\tau\quad \text{s.t.} \nonumber\\[4pt]
&  \begin{aligned}
 &\int_{\T_x}^\infty  (2-\p_n(s(\tau))) d\tau \\[4pt]
 &~~+\int_0^{\T_x} \big( x +  (1-\Q_n(s(\tau)))\big) d\tau \le2~~~~  \forall x\in [0,1] \nonumber \\[6pt]
& \int_0^\infty  \big(1-\Q_n(s(\tau))\big) d\tau ~\le~ 1\\[4pt]
&\text{$s(\cdot)$ is weakly decreasing}\,,~~&\nonumber
\end{aligned}
 \label{lp:esp_n_non_dis} \tag{\small{\sf{FR-ESP-n}}}
\end{align*}
\begin{align*}
&{\text{\sf \small{FR-ESP-d}}}(n,k)~=~ \\
& \max_{w} \sum_{i\in [k]} w_i \quad \text{s.t.} \nonumber \\
&\begin{aligned}   &\sum_{i=1}^j  w_i \frac{2-\p_n(\mathtt{s}_i)}{\mathtt{s}_i} \\ 
\quad &~~~+  \sum_{i=j+1}^{k} w_i \frac{ \mathtt{s}_j + (1-\Q_n(\mathtt{s}_i))}{\mathtt{s}_i}
    ~\leq~ 2,  &\forall j \in [k] \\
   & \sum_{i=1}^k w_i \frac{1-\Q_n(\mathtt{s}_i)}{\mathtt{s}_i} ~\leq ~1 & \\
     & w_i \geq 0\,. &\forall i \in [k]
\end{aligned}
\label{lp:esp_n} \tag{\small{\sf FR-ESP-n-d}}
\end{align*}
\end{multicols}
	 
Here, for any $x\in [0,1]$,  $\T_x~=~  
\inf\{\tau: s^\star(\tau) \le  x\}$, $\mathtt{s}_i = i/k$, $i\in [k]$,  $\Q_n(y) = \left( 1 - \frac{y}{n} \right)^n $, and $\p_n(y) =  2 \left( 1 - \frac{y}{n} \right)^n + y
\left( 1 - \frac{y}{n} \right)^{n-1}$. 
\end{theorem}

\medskip

{\color{black}For $n 
\in [10]$ and  $k \in \{200, 400, 800, 1600\}$, the approximation factor of  Theorem \ref{thm:esp_n}, i.e., $\frac{1}{{\text{\sf FR-ESP-d}}(n,k)}$, is presented in Table  
\ref{table:esp_finite_n}.} The proof of Theorem \ref{thm:esp_n} is similar to the proof of Theorem \ref{thm:pp2}; thus, it is omitted. The only difference between the proofs is that here we provide tighter lower bounds for the revenue of the Myersonian and uniform ESP auctions using Lemma \ref{lem:prob}. This lemma provides \textit{n}-independent and \textit{n}-dependent bounds. The \textit{n}-independent bounds---that is, $2e^{-\mathtt{s}_i} +
 \mathtt{s}_ie^{-\mathtt{s}_i}$ and  $e^{-\mathtt{s}_i}$---were used in Theorem \ref{thm:pp2}, while the \textit{n}-dependent bounds---that is, $\p_n(\mathtt{s}_i)$ and $\Q_n(\mathtt{s}_i)$---are used in Theorem \ref{thm:esp_n} to obtain an improved approximation factor (see Equations~\eqref{eq:z_0} and~\eqref{eq:z_0_z_1} in Lemma~\ref{lem:prob} to see how these quantities relate).  Observe that if in Problem \eqref{lp:esp_n}, we  replace $\p_n(\mathtt{s}_i)$ and $\Q_n(\mathtt{s}_i)$, respectively, with
$2e^{-\mathtt{s}_i} +
 \mathtt{s}_ie^{-\mathtt{s}_i}$ and  $e^{-\mathtt{s}_i}$, we recover 
 Problem \eqref{lp:esp:d}. 

\newpage

\medskip
\ECHead{Appendix}
\section{Other Proofs}
\subsection{Proof of Lemma  \ref{lem:prob}}\label{sec:proof:lem:prob}
\begin{oneshot}{\textbf{Lemma~\ref{lem:prob}}}
Let $Z_{\tau}$ be the number of buyers with $v_i ~\geq ~ t_i' ~\ge ~ \tau$; 
that is,  $Z_{\tau}= \sum_{i=1}^n\ind (v_i 
~\geq~  t_i' ~\geq ~ {\tau})$.
Then,
 \begin{align} \P[Z_{\tau} = 0] ~&\le~ \Q_n(s^{\star}(\tau))~\le~ \lim_{n\rightarrow 
 \infty}\Q_n(s^{\star}(\tau)) ~=~ e^{-s^{\star}(\tau)}
\tag{\ref{eq:z_0}} 
\\
 2 \P[Z_{\tau} = 0] + \P[Z_{\tau} = 1] ~&\le~\p_n(s^{\star}(\tau))~\le~ 
 \lim_{n\rightarrow 
  \infty}\p_n(s^{\star}(\tau)) ~=~ (2 + s^{\star}(\tau)) e^{-s^{\star}(\tau)}\,,
\tag{\ref{eq:z_0_z_1}} 
 \end{align}
 where $\Q_n(y) = \left( 1 - \frac{y}{n} \right)^n $ and $\p_n(y) = 2 \left( 1 
 - \frac{y}{n} \right)^n + y
\left( 1 - \frac{y}{n} \right)^{n-1} $.
\end{oneshot}
\medskip
\begin{proof}{{Proof of Lemma~\ref{lem:prob}}} 
Define $z_i~=~\ind(v_i ~\geq ~ t_i' ~ \geq ~ \tau)$. Then,
 $Z_{\tau}$ can be written in the following manner:
$Z_{\tau} = \sum_{i\in [n]}z_i\,,$
where $z_i$'s  are independent $0/1$ Bernoulli random variables with $\E[z_i] = 
s_i(\tau)$. This implies that $\E[Z_{\tau}] = \sum_{i\in [n]} \E[z_i] = \sum_{i 
\in [n]} s_i(\tau) =  s(\tau)$. Then,
\[\P[Z_{\tau} =0] ~=~ \prod_{i\in [n]} \P[z_i = 0]~=~\prod_{i\in [n]} (1-s_{i}(\tau)) ~\le~ \left( 1 - \frac{\sum_{i\in [n]}s_{i}(\tau)}{n} \right)^n ~\le~
 e^{-\sum_{i\in [n]} s_i(\tau)} ~=~ e^{-s(\tau)}\,,\]
 where the first inequality follows from the fact that  for any sequence $a_1, a_2, \ldots, a_n$, we have $\prod_{i\in [n]} a_i \le \big(\frac{\sum_{i\in [n]} a_i}{n}\big)^n$. By definition of $\Q_n(\cdot)$, the above equation leads to {Inequality (\ref{eq:z_0}),} which is the first desired result. 
 
Next, we show Inequality  (\ref{eq:z_0_z_1}). That is, we show   $2 \P[Z_{\tau} = 0] + \P[Z_{\tau} = 1] ~\le~\p_n(s(\tau))~\le~ (2 + s(\tau)) 
 e^{-s(\tau)}$. We begin by observing that the l.h.s. of this equation 
can be written as a symmetric polynomial in $s_1(\tau), \hdots, s_n(\tau)$, namely,
\[2 \P[Z_{\tau} = 0] + \P[Z_{\tau} = 1] ~=~  2 \prod_{i\in [n]} (1-s_i(\tau)) + \sum_{i\in [n]} s_i(\tau)
\prod_{j \neq i} (1-s_j(\tau))\,.\]
In the remainder of the proof, to ease the notation, we denote $s_i(\tau)$, $i\in [n]$, by $s_i$.
Define polynomial $ P_n(s_1, \ldots,s_n) ~:=~ 2 \prod_{i\in [n]} (1-s_i) + \sum_{i\in [n]} s_i
\prod_{j \neq i} (1-s_j)$. To provide an upper bound on $2 \P[Z_{\tau} = 0] + 
\P[Z_{\tau} = 1]$, we show that subject to the constraint $\sum_{i\in [n]} s_i 
= s(\tau)$,
the value of the polynomial $ P_n$ is maximized when $s_1=s_2=\dots=s_n = s(\tau)/n$. 

In order to prove this, consider a point $\bs=(s_1,\dots, s_n)$, such that $\sum_{i\in [n]} s_i = s(\tau)$. Pick any pair of coordinates (without loss of generality, $1$ and $2$) and consider increasing one and decreasing the other. Now, note that $ P(s_1+\delta, s_2-\delta, s_3, \dots,s_n)$ is a quadratic function of $\delta$. It is not difficult to verify that the quadratic coefficient is negative. Then, considering the fact that $ P$ is symmetric, it follows that the maximum in this direction is achieved at $\delta$, such that $s_1+\delta = s_2-\delta$---that is, when the two coordinates are equal, for \emph{every} profile of values for the remaining coordinates. Since this argument holds for any pair of coordinates, it follows that $ P(\bs)$ is maximized when all coordinates are equal---that is,  $s_i = s(\tau)/n$ for
$i\in [n]$. 

Thus far, we have established that 
\[2 \P[Z_{\tau} = 0] + \P[Z_{\tau} = 1] ~\le~  P_n\Big(\frac{s(\tau)}{n}, \ldots,\frac{s(\tau)}{n}\Big) ~=~  \p_n(s(\tau))\,,  \]
where the equality follows from the definitions of $P_n$ and $\p_n$. The above equation yields the first desired inequality in (\ref{eq:z_0_z_1}). For the second inequality, we observe that
$$ P_n\Big( \frac{s(\tau)}{n}, \hdots, \frac{s(\tau)}{n} \Big) ~=~ 
 P_{n+1}\Big( \frac{s(\tau)}{n}, \hdots, \frac{s(\tau)}{n}, 0 \Big) ~\leq~
 P_{n+1}\Big( \frac{s(\tau)}{n+1}, \hdots, \frac{s(\tau)}{n+1}, \frac{s(\tau)}{n+1} \Big)\,.  $$
In particular,
$$ P_n\Big( \frac{s(\tau)}{n}, \hdots, \frac{s(\tau)}{n} \Big)  ~\leq~ \lim_{k
\rightarrow \infty}   P_k\Big( \frac{s(\tau)}{k}, \hdots, \frac{s(\tau)}{k}
\Big) ~=~ (2 + s(\tau)) e^{-s(\tau)}\,.$$
$\blacksquare$
\end{proof}

\subsection{Proof of Discretized Bound of Theorem \ref{thm:pp_n}}\label{sec:proof:thm:pp_n-dis}
\begin{oneshot}{\textbf{Theorem~\ref{thm:pp_n} ($n$-Dependent Revenue Bound of SPP Mechanisms in Single-unit Settings).}}
In a $1$-unit $n$-buyer setting with independent  private values, there {exists a vector of prices} $\mathtt{\bp} =(\mathtt{p}_1, \mathtt{p}_2, \ldots, \mathtt{p}_n)$ such that 
\begin{itemize}
\item 
	 \textbf{Non-discretized Bound.}   $\PP_n(\bp) \geq \Opt_n\cdot \frac{1}{{{\text{\sf FR}}(n)}}$, and 
	 \item \textbf{Discretized  Bound.} 
	 $\PP_n(\bp) \geq \Opt_n\cdot \frac{1}{{\text{\sf FR-d}}(n,k)}$ for any positive integer $k$,
\end{itemize}
 where $\Opt_n$ is the expected optimal revenue in a $1$-unit setting,  $\PP_n(\bp)$ is the expected revenue of the SPP mechanism with prices $\mathtt{\bp}$, and $\text{\small{\sf{FR}}}(n)$ and ${{\text{\sf FR-d}}(n,k)}$ are defined as
 
\begin{multicols}{2}\setlength{\columnseprule}{1pt}
\noindent
  \begin{align*}
  &\text{\small{\sf{FR}}}(n) ~=~
  \\   
&  \begin{aligned}
\quad \max_{\{{s(\tau), \tau\ge 0\}}}& ~~\int_0^\infty 
s(\tau) d\tau\\
\text{s.t.} \quad &0~ \le s(\tau) \le \min(1, 1/\tau)~ & &\forall\tau \ge 0   \\[8pt]
& \int_0^\infty  \big(1-\Q_n(s(\tau))\big) d\tau ~\le~ 1 \\[8pt]
&\text{$s(\cdot)$ is weakly decreasing}\,.&
\end{aligned}
 \label{lp:spm_n_1} \tag{{{\small\sf{FR-n}}}}
\end{align*} 
\begin{align*}
&{\text{\small \sf FR-d}}(n,k)~=~ \\
  & \quad\begin{aligned}
 \max_{w} &
\sum_{i\in [k]} w_i & &\\
  \text{s.t.}\quad&{\sum_{i =j+1}^{k}} w_i \frac{ \mathtt{s}_j }{\mathtt{s}_i} ~\leq~ 1 & \forall j \in [k-1]& \\
  &\sum_{i=1}^k w_i \frac{1-\Q_n(\mathtt{s}_i)}{\mathtt{s}_i} ~\leq ~1&& \\
    & w_i ~\geq~ 0. &\forall i \in [k]&&
\label{lp:spm_n_dist} 
\end{aligned}
\tag{{\sf {\small{FR-n-d}}}}
  \end{align*}
\end{multicols}
Here, $\Q_n(y) = \left( 1 - \frac{y}{n} \right)^n $ and  $\mathtt{s}_i = i/k$, $i\in [k]$. Further, for $n\in [10]$ and $k\in\{200, 400, 800, 1600\}$,  our approximation factors of $\frac{1}{{{\text{\sf FR-d}}(n,k)}}$ are presented in Table \ref{table:pp_finite_n}.
 \end{oneshot}

\medskip
\begin{proof}{{Proof of Discretized Bound of Theorem~\ref{thm:pp_n}}}
{To show the result, we verify that $w_{i}^\star$s, defined in Equation (\ref{eq:w}), satisfy the constraints of Problem \eqref{lp:spm_n_dist}. Similar to the proof of Theorem \ref{thm:pp1}, we normalize the revenue of the  SPP mechanism that selects the best of the Myersonian and uniform pricing rules to one; that is, $\max (\MP, \UP) =1$.}

\textbf{First Set of Constraints.} {Here, we show that $w_{i}^\star$'s satisfy the first set of constraints.}  
Define $\T_x~=~  \inf\{\tau: s^\star(\tau) \le  x\}$, $x\in [0,1]$. With a slight abuse of notation, let $\UP_x$ be the revenue of the SPP mechanism that posts a uniform price of $\T_x$ for all 
buyers. By definition of the uniform SPP mechanism, we have $\UP~\ge~ \max_{x\in [0,1]} ~ \UP_x$.   
 Define $u_x(\tau)$ as the probability that the SPP mechanism with uniform price $\T_x$ sells with a price of at least
$\tau$. Then, $\UP_x = \int_{\tau =0}^{\infty} u_x(\tau) d\tau$.  
Next, we bound $\UP_x$ by bounding $u_x(\tau)$. For $\tau \leq \T_x$, we bound $u_x(\tau)$ by 
\begin{align}u_x(\tau) ~ \geq~ s^\star(\T_x)~ \geq ~ x\,,\quad \tau \le \T_x\,. \label{eq:u_x_1_1}\end{align}
 This bound holds because (i) while the SPP mechanism with uniform price $\T_x$ can sell the item with a price of at least
$\tau$ if there exists at least one buyer $i$ with value $v_i \geq \T_x$, the 
optimal mechanism can sell at a price of at least $\T_x$ only if there is at least one buyer $i$ with $v_i~ \ge ~ t_i ~\ge ~ \T_x$. Consequently, $u_x(\tau) ~ \geq~ s^\star(\T_x)$, and (ii) by definition of $\T_x$,  we 
have $s^\star(\T_x)~ \geq ~ x$; to see this, recall that $\T_x~=~  
\inf\{\tau: s^\star(\tau) \le  x\}$. Thus, when $\T_x\in \{\tau: s^\star(\tau) \le  x\}$, by 
monotonicity  of $s^\star(\tau)$, it must be the case that $s^\star(\T_x) =  x$. Further, 
if $\T_x\notin \{\tau: s^\star(\tau) \le  x\}$, we have $s^\star(\T_x) >  x$. Thus, 
$s^\star(\T_x)~ \geq ~ x$. Then, by {Inequality (\ref{eq:u_x_1_1}),} and our assumption that $\max (\MP, \UP)= 1$, we have 
\[1 ~\ge~ \UP_x ~ \ge~  \int_0^{\T_x} x d\tau ~=~  \int_0^{\T_x} \frac{x}{s^\star(\tau)} s^\star(\tau) d\tau\,. \]
In the following, we set $x$ to $\mathtt{s}_j = j/k$.  Then, we obtain
{\[ 1~\ge ~\int_0^{\T_x} \frac{\mathtt{s}_j}{s^\star(\tau)} s^\star(\tau) d\tau  ~= ~ \sum_{i=j+1}^{k} \int_{\uptau_{i}}^{\uptau_{i-1}}  \frac{\mathtt{s}_j}{s^\star(\tau)} s^\star(\tau) d\tau~\ge  \sum_{i =j+1}^{k} w_{i}^\star \frac{ \mathtt{s}_j }{\mathtt{s}_i}\,,\]} 
\noindent where the first equality follows from the definitions of $\uptau_i$'s and $\T_x$, and the second inequality follows from the definition of $w_{i}^\star$ and  the fact that 
$s^\star(\cdot)$ is weakly decreasing.  (Recall that $0 = \uptau_k \leq \uptau_{k-1} \leq \ldots \leq \uptau_1 \leq
\uptau_0 = \infty$ such that $\uptau_j =\inf\{\tau: s^\star(\tau) \le  j/k\}$, $j\in [k-1]$, and  $\Opt ~=~ \sum_{i\in [k]} w^\star_{i}$, where $ w_{i}^\star ~=~ \int_{\uptau_{i}}^{\uptau_{i-1}} s^\star(\tau) d\tau$.) Note that the above equation verifies the first set of constraints. 

\textbf{Second Constraint.}  {Here, we show that $w_{i}^\star$'s satisfy the second set of constraints.} Let $\m(\tau)$ be the probability
that the Myersonian SPP mechanism sells with a price of at least $\tau$. Then, by construction of the  prices, $t_i'$'s, in this mechanism, we obtain 
$\m(\tau)~ = ~1-\P[Z_{\tau} = 0]$, 
where $Z_{\tau}$ is the number
of buyers who satisfy  $v_i ~\geq ~ t_i' ~ \geq ~ \tau$. This implies that 
{\begin{align}\nonumber 1~&\ge~\MP~= ~ \int_{0}^{\infty} (1-\P[Z_{\tau} = 0]) 
d\tau 
~\ge ~ \int_{0}^{\infty} (1-\Q_n(s(\tau)))  d\tau \\\nonumber
~&=~  \int_{0}^{\infty} \frac{(1-\Q_n(s(\tau)))}{s(\tau)} s(\tau)  d\tau ~\ge~ \sum_{i=1}^k w_{i}^\star \frac{1-\Q_n(\mathtt{s}_i)}{\mathtt{s}_i}\,, \end{align}}
where the second inequality follows from Lemma \ref{lem:prob}, and third inequality follows from Lemma \ref{lem:monotone_1}, where we show that $\frac{1-\Q_n(y)}{y}$ is decreasing in $y$. The above equation confirms that $w_i^{\star}$s satisfy the second constraint of Problem \eqref{lp:spm_n_dist} and completes the proof.
$\blacksquare$
\end{proof}

\begin{lemma} \label{lem:monotone_1}
  Function $(x,y) \mapsto \frac{1}{y} (x + 1-\Q_n(y))$ 
is decreasing in $y \in [0,1]$ for every 
positive integer $n$ and every $x \geq 0$, 
where $\Q_n(y) = (1-\frac{y}{n})^n$. 
\end{lemma}
\begin{proof}{Proof of Lemma \ref{lem:monotone_1}}
The derivative of this function w.r.t. $y$ is given by 
\begin{align*}
\frac{\partial\left( \frac{1}{y} (x + 1-\Q_n(y))\right)}{\partial y}~=~ \dfrac{-x-1+y(1-\frac{y}{n})^{n-1}+(1-\frac{y}{n})^n}{y^2} = 
\dfrac{(1-\frac{y}{n})^{n-1}\left( -\frac{x+1}{(1-\frac{y}{n})^{n-1}}+y+1-\frac{y}{n} \right)}{y^2}\,.
\end{align*}
To show that  $\frac{\partial\left( \frac{1}{y} (x + 
1-\Q_n(y))\right)}{\partial y}~\le~ 0$, we verify that 
$-\frac{x+1}{(1-\frac{y}{n})^{n-1}}+y+1-\frac{y}{n} ~\le ~0$.
For $y < 1$ and $x \ge 0$, we have
\begin{align*}
\dfrac{x+1}{(1-\frac{y}{n})^{n-1}}&~\ge~ (x+1) \left( 1+\frac{y}{n}\right)^{n-1}~\ge~ (x+1)\left(1+\frac{n-1}{n}y  \right)~\ge~ 1+\frac{n-1}{n}y\,.
\end{align*}
The last inequality implies that $\frac{1}{y} (x + 1-\Q_n(y))$ is decreasing in 
$y$.
$\blacksquare$
\end{proof} 

\subsection{{Proof of Theorem \ref{thm:pp_multi}}} \label{sec:proof_multi}
\begin{oneshot}
{\textbf{Theorem~\ref{thm:pp_multi} {(Revenue Bound of SPP Mechanisms in Multi-unit Settings).}}}
 In an $\hu$-unit $n$-buyer setting with independent  private values, there {exists a vector of prices} $\mathtt{\bp} =(\mathtt{p}_1, \mathtt{p}_2, \ldots, \mathtt{p}_n)$ such that $
	 \PP_{\hu}(\bp) \geq \Opt_{\hu}\cdot \frac{1}{\text{\ref{lp:spmk}}}$, where $\Opt_{\hu}$ is the expected optimal revenue in an $\hu$-unit setting,  $\PP_{\hu}(\bp)$ is the expected revenue of the SPP mechanism with prices $\mathtt{\bp}$, and {\sf{\small {FR-Multi}}({\hu})} is: 
  \begin{align} 
{\small{\text{\sf{\small FR-Multi}}({\hu})}} ~=~&{\max_{\{{s(\tau), \tau\ge 0\}}}} ~~\int_0^\infty 
s(\tau) d\tau \nonumber  \\
&  \begin{aligned}
\text{s.t.} ~~~ &0~ \le~ s(\tau) ~\le~ \min({\hu}, 1/\tau)~~ & &\forall ~~ \tau \ge 0 \nonumber \\
& \int_0^\infty  f_{\hu}(s(\tau)) d\tau ~\le~ 1\\
&\text{$s(\cdot)$ is weakly decreasing}\,.~~&\nonumber
\end{aligned}
\label{lp:spmk} 
 \tag{{{\small{\sf{FR-Multi({\hu})}}}}}
\end{align}
Here,  ${f_{\hu}(x) =  {\hu}- e^{-x} \sum_{i=0}^{{\hu}-1} ({\hu}-i)\frac{x^i}{i!} }$. Our bound is greater than the best-known bound prior to this work---that is, $\frac{1}{{\text{\ref{lp:spmk}}}} > 1- \frac{{\hu}^{\hu}}{{\hu}! e^{\hu}}
$. 
\end{oneshot}
\medskip

\begin{proof}{Proof of Theorem \ref{thm:pp_multi}}
{In the first part of the proof, we show that the SPP mechanism with the best of the Myersonian and uniform prices in  ${\hu}$-unit settings yields the desired bound. In the second part of the proof, we show that our bound outperforms the best-known bound prior to this study. 
 
 \textbf{First Part (Showing the Bound).} We begin by revisiting the definition of Myersonian and uniform prices for the SPP mechanism.}  

 \textbf{{Myersonian SPP Mechanism.}} {Approach the buyers in decreasing  order of their Myersonian prices---that is, the re-sampled thresholds $t_i'$ (defined in Section~\ref{prelim:defn}), and {allocate to} the first ${\hu}$ buyers whose values $v_i$ exceeds their threshold $t_i'$. With a slight abuse of notation, let $\MP$ denote the expected revenue of this mechanism, where the expectation {is taken w.r.t.} the randomness in both the re-sampled posted prices and the buyers' values.}
 
\textbf{{Uniform SPP Mechanism.}} {Approach buyers in an arbitrary order, and 
 {allocate to} the first $\hu$ buyers whose value exceeds the price $p^{\star} = 
 \argmax_p p\cdot \sum_{i=1}^{{\hu}} \P[v_{(i)} \geq p]$, where $v_{(i)}$ is the $i$-th highest value. Equivalently, $p^{\star} = 
 \argmax_p p\cdot \E[\min(|S_p(\mathbf{v})|, \hu)]$, where $S_p(\mathbf v)$ is the set of buyers with $v_i \geq p$. With a slight abuse of notation, let $\UP$ be the 
 expected revenue of this mechanism, where the expectation is taken with respect to the buyers' values.}

{ 
   As usual, without loss of generality, we assume that $\max(\MP, \UP) =1$ and show that $\Opt_{\hu} \le \max(\MP, \UP)  \cdot {\text{\ref{lp:spmk}}} $. We prove this result by showing that the function $s^\star(\cdot)$  corresponding to the optimal mechanism is a feasible solution to Problem \eqref{lp:spmk}. 
   } 

\textbf{Lower Bounds on $\boldsymbol \UP$ (First Set of Constraints).} 
The revenue of the SPP mechanism that posts a price of $\tau$ for every buyer is equal to $\tau \cdot  \sum_{i=1}^{{\hu}} \P[v_{(i)} \geq \tau]$, which is at least  $\tau s^\star(\tau)$. Therefore, $\UP \geq \tau s^\star(\tau)$ for every $\tau \geq 0$---that is, 
\begin{align}\max_{\tau \ge 0} ~\tau s^\star(\tau) ~\le~ \UP~\le~ 1\,, \label{eq:bound_UP_2}  \end{align}
where the second inequality follows from $\max(\MP, \UP) =1$. Equation 
(\ref{eq:bound_UP_2}) results in
\begin{align} \label{eq:st_ineq_2}
&0 ~\le~ s^\star({\tau})~ \le~ \min({\hu}, 1/{\tau})~ \quad\forall\tau\ge 0\,.
\end{align}
In the inequality, we also use the fact that  $s^\star(\tau)$  is at most ${\hu}$. {Note that Equation (\ref{eq:st_ineq_2}) is the first set of 
constraints in Problem \eqref{lp:spmk}.} 

\textbf{Lower Bounds on $\boldsymbol \MP$ (Second Constraint).}  We define $m(\tau)$ as 
the expected number of units that the Myersonian  SPP mechanism
sells with a price of
at least ${\tau}$. This yields
$\textstyle \MP ~=~ \int_0^\infty \m ({\tau}) d{\tau} ~\le ~1$, 
where the inequality follows from $\max(\MP, \UP ) =1$. 
Next, we present a lower bound on $\MP$. Let $Z_{\tau}$ be the number of buyers with  $v_i 
~\geq~  t_i' ~\geq ~ {\tau}$; that is, $Z_{\tau}= \sum_{i=1}^n\ind (v_i 
~\geq~  t_i' ~\geq ~ {\tau})$. 
 Then,
\[m(\tau) ~=~ \sum_{i =1}^{{\hu}-1} i\P[Z_{\tau} = i] +{\hu} \P[Z_{\tau}\ge {\hu}] ~=~\sum_{i =1}^{{\hu}-1} i\P[Z_{\tau} = i] +{\hu} (1-\sum_{i=0}^{{\hu}-1}\P[Z_{\tau}= i])\,.\]
This leads to $m(\tau) ={\hu}-  \sum_{i=0}^{{\hu}-1} ({\hu}-i) \P[Z_{\tau} =i]$. 
Invoking Lemma \ref{lm:equal}, we obtain
\begin{align}m(\tau) &~=~{\hu}-  \sum_{i=0}^{{\hu}-1} ({\hu}-i)\left( \begin{array}{c}
n  \\
i \end{array} \right)\frac{s^\star(\tau)^i}{n^i}(1-\frac{s^\star(\tau)}{n})^{n-i} ~ \ge~ {\hu}- e^{-s^\star(\tau)} \sum_{i=0}^{{\hu}-1} ({\hu}-i)\left( \begin{array}{c}
n  \\
i \end{array} \right)\frac{s^\star(\tau)^i}{n^i} =  f_{\hu}(s^\star (\tau))\,,\nonumber 
\end{align}
where  
$f_{\hu}(x) = {{\hu}- e^{-x} \sum_{i=0}^{{\hu}-1} ({\hu}-i)\frac{x^i}{i!} }$.\medskip

{ \textbf{Second Part (Beating the Best-known Bound).} We first invoke Lemma \ref{lem:LP_multi} to write  \ref{lp:spmk} in the following manner: $1 + \ln({\hu}\tau^{\star})$, where $\tau^*> \frac{1}{{\hu}}$ is the unique solution of the following equation:
 \begin{align}
\int_{1/{\hu}}^{\tau^{\star}}  \left({\hu}- e^{-1/\tau} \sum_{i=0}^{{\hu}-1} \frac{({\hu}-i)}{\tau ^i i!}\right)  d\tau ~=~ \frac{{\hu}^{\hu}}{{\hu}! e^{\hu}}\,. \label{eq:char_tau}\end{align}
To show the result, we verify that 
$\ln({\hu} \tau^*)<\frac{{\hu}^{\hu}}{{\hu}! e^{\hu}}\,.$
This ensures that $\frac{1}{{\text{\ref{lp:spmk}}}} = \frac{1}{1+\ln({\hu} \tau^*)} > 1- \frac{{\hu}^{\hu}}{{\hu}! e^{\hu}}
$, which is the desired result.
We begin by simplifying the summation in the l.h.s. of Equation (\ref{eq:char_tau}). Observe that 
 \begin{align*}
{\hu}\sum_{i=0}^{{\hu}-1} \frac{1}{\tau ^i i!}  = {\hu}e^{\frac{1}{\tau}}- {\hu} \sum_{i={\hu}}^{\infty} \frac{1}{\tau ^i i!} \,,
\end{align*}
and 
 \begin{align*}
\sum_{i=0}^{{\hu}-1} \frac{i}{\tau ^i i!}  = \sum_{i=1}^{{\hu}-1} \frac{i}{\tau ^i i!} = \sum_{i=1}^{{\hu}-1} \frac{1}{\tau ^i (i-1)!} = \frac{1}{\tau}\sum_{i=0}^{{\hu}-2} \frac{1}{\tau ^i i!} =   \frac{1}{\tau}\left(e^{\frac{1}{\tau}}-\sum_{i={\hu}-1}^{\infty} \frac{1}{\tau ^i i!}\right) \,.
\end{align*}
Having simplified the summations, we now revisit Equation (\ref{eq:char_tau}):
\begin{align*}
\int_{1/{\hu}}^{\tau^{\star}}  \left({\hu}- e^{-\frac{1}{\tau}} \left(({\hu}-\frac{1}{\tau})e^{\frac{1}{\tau}}- ({\hu}-\frac{1}{\tau}) \sum_{i={\hu}}^{\infty} \frac{1}{\tau ^i i!} +  \frac{1}{\tau ^{{\hu}} ({\hu}-1)!}\right)\right)  d\tau = \frac{{\hu}^{\hu}}{{\hu}! e^{\hu}}  \\
\Rightarrow ~~ \ln({\hu}\tau^*) +\int_{1/{\hu}}^{\tau^{\star}}   \left(e^{-\frac{1}{\tau}}({\hu}-\frac{1}{\tau}) \sum_{i={\hu}}^{\infty} \frac{1}{\tau ^i i!} +  \frac{e^{-1/\tau}}{\tau ^{{\hu}} ({\hu}-1)!}\right)  d\tau = \frac{{\hu}^{\hu}}{{\hu}! e^{\hu}}
\,.
\end{align*}
Note that the integral in the l.h.s. of the above equation is positive. This implies that $\ln({\hu}\tau^*)<  \frac{{\hu}^{\hu}}{{\hu}! e^{\hu}}$, which is the desired result.}
$\blacksquare$
\end{proof}

\subsection{{Proof of Lemma \ref{lm:equal}} } \label{sec:equal}

\begin{oneshot} {\textbf{Lemma~\ref{lm:equal} (Lower Bound of Myersonian SPP Mechanisms in Multi-unit Settings).}}
Consider the $\hu$-unit setting. 
Let ${s^\star(\tau )} = \sum_{i =1}^n \P[ v_i~ \geq ~t_i ~\geq ~{\tau}]$.  
Then, $m(\tau)$, which is the expected number of units that the Myersonian SPP mechanism sells with a price of at least ${\tau}$, satisfies the following inequality.
$$m(\tau) \geq \hu- \sum_{i=0}^{\hu-1} (\hu-i)\left( \begin{array}{c}
n  \\
i \end{array} \right)\frac{s^\star(\tau)^i}{n^i}(1-\frac{s^\star(\tau)}{n})^{n-i}.$$
\end{oneshot}

\begin{proof}{Proof of Lemma \ref{lm:equal}} Let $m(\tau)$ be the expected number of units that the Myersonian SPP mechanism sells with a price of at least ${\tau}$. As earlier, we define $Z_{\tau}$ as the number of buyers with $v_i 
~\geq~ t_i' ~\geq ~ {\tau}$---that is, $Z_{\tau}= \sum_{i=1}^n\ind (v_i 
~\geq~  t_i' ~\geq ~ {\tau})$. Then, \[m(\tau) ~=~ \sum_{i =1}^{{\hu}-1} i\P[Z_{\tau} = i] +{\hu} \cdot\P[Z_{\tau}\ge {\hu}] ~=~\sum_{i =1}^{{\hu}-1} i\P[Z_{\tau} = i] +{\hu} (1-\sum_{i=0}^{{\hu}-1}\P[Z_{\tau}= i])\,.\]
Note that the second term of $m(\tau)$ is ${\hu}\cdot \P[Z_{\tau}\ge {\hu}]$ because we cannot serve more than $\hu$ buyers. Collecting common terms yields
\begin{align}m(\tau) ={\hu}- \sum_{i=0}^{{\hu}-1} ({\hu}-i) \P[Z_{\tau} =i]\,. \label{eq:m}\end{align}
 Given this, we begin with writing $\sum_{i=0}^{{\hu}-1}({\hu}-i)\P[Z_{\tau}=i]$ as a function of $s^\star_i (\tau)$, $i\in [n]$, where $s^\star_i(\tau) = \P[v_i~ \geq ~t_i ~\geq ~{\tau}]=\P[v_i~ \geq ~t_i' ~\geq ~{\tau}]$. 
In the proof, to simplify the notation, we denote $s^\star(\tau)$ and $s^\star_i(\tau)$, $i\in [n]$, with $s$ and $s_i$, respectively. 
Define polynomial $P_n(s_1, \ldots, s_n) := \sum_{i=0}^{{\hu}-1}({\hu}-i)\P[Z_{\tau}=i]$. We find an upper bound on $P_n(s_1, \ldots, s_n)$. By definition, $P_n(s_1, \ldots, s_n)$ is equal to
$${\hu}\prod_{i \in [n]} (1-s_i) + ({\hu}-1) \sum_{i\in [n]} s_i\prod_{j\neq i}(1-s_j)+\ldots +\sum_{S, S \subseteq [n], |S|={\hu}-1} \prod_{i\in S} s_i\prod_{j \in [n]-S}(1-s_j)\,.$$
We show that subject to $\sum_{i \in [n]} s_i = s$, the value of the polynomial $P_n$ is maximized when $s_1=s_2=\ldots=s_n = s/n$. This completes the proof.

To show this, consider a point $s=(s_1, \ldots, s_n)$, such that $\sum_{i \in [n]} s_i = s$. Select any pair of coordinates (without loss of generality, coordinates 1 and 2) and consider increasing one and decreasing the other. We show that 
\begin{align}\label{eq:difference}
P_n(s_1+\delta, s_2-\delta, s_3, \ldots, s_n)-P_n(s_1, s_2, \ldots, s_n)
\end{align} is quadratic in $\delta$, and the quadratic coefficient is negative. Then, considering the fact that $P_n$ is symmetric, it follows that the maximum in this direction is achieved at $\delta$, such that $s_1+\delta = s_2-\delta$---that is, when the two coordinates are equal. Since this argument holds for any pair of coordinates, it follows that the polynomial is maximized when all the coordinates are equal---that is, $s_i=s/n$ for $i \in [n]$. Showing this yields the desired result. 

Note than each term of polynomial $P_n(s_1, \ldots, s_n)$ is either a product of $s_i$, $i\in [n]$, or a product of $(1-s_i)$. For a given term in Equation (\ref{eq:difference}), we say $s_i$, $i \in [n]$, is in the first ``location" if this term is a product of $s_i$ and we say $s_i$ is in the second location if this term is a product of $1-s_i$. Then, we group the terms in expression (\ref{eq:difference}) based on the locations of $s_i$ for $i \in [n]-\{1,2\}$---that is, we put all the terms with the same location for all $i \in [n]-\{1,2\}$ in the same group. Now, consider a certain group. Note that any term in this group can be written as a product of $\prod_{i \in \text{Loc}_1} s_i \prod_{j\in \text{Loc}_2} (1-s_j)$ for $i,j \neq 1,2$, where Loc$_1$ and Loc$_2$ are the subsets of indices that are in the first and second locations, respectively, in the aforementioned group. Specifically, Loc$_2 = [n]-(\text{Loc}_1 \cup \{1, 2\})$. 
 Let us call this $\prod_{i \in \text{Loc}_1} s_i \prod_{j\in \text{Loc}_2} (1-s_j)$ a common sub-term of the group. We are interested in the multiplier of the common sub-term in $P_n(s_1+\delta, s_2-\delta, s_3, \ldots, s_n)-P_n(s_1, s_2, \ldots, s_n)$. 
 We will show that the multiplier of the common sub-term in $P_n(s_1+\delta, s_2-\delta, s_3, \ldots, s_n)-P_n(s_1, s_2, \ldots, s_n)$ is always zero, unless $L := |\loc_1| = \hu -1$. Further, we show that any non-zero multiplier of the common sub-term is quadratic and concave in $\delta$.
 
We consider the following three cases.
 
 \begin{itemize}
 \item Case 1 ($L \le \hu-3$): The multiplier of the common sub-term in $P_n(s_1+\delta, s_2-\delta, s_3, \ldots, s_n)-P_n(s_1, s_2, \ldots, s_n)$ depends on the location of $s_1$ and $s_2$ in $P_n$. Thus, the group associated with the common sub-term has four members, where each member corresponds to one particular location for $s_1$ and $s_2$. We consider each of these members separately.  
 \begin{itemize}
 \item Member 1: Both $s_1$ and $s_2$ are in the first location. In this case, the multiplier of the common sub-term is 
 \[({\hu}-(L+2))\cdot \big((s_1+\delta)(s_2-\delta) - s_1 s_2\big)=({\hu}-(L+2)) \cdot \left(-\delta s_1+\delta s_2- \delta^2\right)\,.\]
 The multiplier $({\hu}-(L+2))$ is due to the fact that the number $s_i$'s, $i\in [n]$, which is in the first location is $|\loc_1|+ 2 = L+2$.
 \item Member 2: Both $s_1$ and $s_2$ are in the second location. In this case, the multiplier of the common sub-term is 
  \begin{align}({\hu}-L)\cdot \big((1-s_1-\delta)(1-s_2+\delta) - (1-s_1) (1-s_2)\big)=   ({\hu}-L) \cdot \left(-\delta s_1+\delta s_2- \delta^2\right)\,. \label{eq:member_2}\end{align}
  \item Members 3 and 4: One of $s_i$, $i\in \{1, 2\}$, is in the first location and the other one is in the second location. In this case, the multiplier of the common sub-term is 
  \begin{align}\nonumber &({\hu}-(L+1))\big( (s_1+\delta)(1-s_2+\delta) - s_1(1-s_2) \big) \\
  &+({\hu}-(L+1))\big( (1-s_1-\delta)(s_2-\delta) - (1-s_1)s_2 \big)  = 2({\hu}-(L+1)) \big(\delta s_1-\delta s_2+ \delta^2\big)\,.\label{eq:member}\end{align}
 \end{itemize}
 Putting all these together, it is easy to see that the multiplier of the common sub-term with $|\loc_1|\le \hu-3$ is zero.
  \item Case 2: ($L = \hu-2$): Here, the group associated with the common sub-term has three members. Note that both $s_1$ and $s_2$ cannot be in the first location as the number of $s_i$s in the first location cannot exceed $\hu-1$.
   \begin{itemize}
 \item Member 1: Both $s_1$ and $s_2$ are in the second location. In this case, the multiplier of the common sub-term is 
  \[({\hu}-L)\cdot \big((1-s_1-\delta)(1-s_2+\delta) - (1-s_1) (1-s_2)\big)=   ({\hu}-L) \cdot \left(-\delta s_1+\delta s_2- \delta^2\right) = 2\left(-\delta s_1+\delta s_2- \delta^2\right)\,, \]
 where the last equality follows because $L = \hu-2$. 
 \item Members 2 and 3: One of $s_i$, $i\in \{1, 2\}$, is the first location and the other one is in the second location. In this case, by Equation (\ref{eq:member}), the multiplier of the common sub-term is 
  \begin{align*} 2({\hu}-(L+1)) \big(\delta s_1-\delta s_2+ \delta^2\big) =2\big(\delta s_1-\delta s_2+ \delta^2\big)\,, \end{align*}
 where the equation holds because $L = \hu-2$. 
 \end{itemize}
 Considering this, it is evident that the multiplier of the common sub-term with $|\loc_1|= \hu-2$ is zero.
\item Case 3: $L={\hu}-1$: In this case, the group has only one member for which both $s_1$ and $s_2$ are in the second location. Thus, by Equation (\ref{eq:member_2}), the coefficient of the common multiplier is $\left(-\delta s_1+\delta s_2- \delta^2\right)$. Observe that this term is quadratic and concave in $\delta$. This observation completes the proof.
  \end{itemize}
$\blacksquare$
\end{proof}
\subsection{{Proof of Lemma~\ref{lem:LP_multi}}} \label{sec:proof:lem:LP_multi}
 {
\begin{oneshot}
{\textbf{Lemma~\ref{lem:LP_multi} ({\color{black}Characterization} {\text{\ref{lp:spmk}}}).}}  Consider any positive integer $\hu>1$.  
 Let $\tau^{\star}> \frac{1}{{\hu}}$ be the unique solution of the following equation
 \begin{align*}
\int_{1/{\hu}}^{\tau^{\star}}  \left({\hu}- e^{-1/\tau} \sum_{i=0}^{{\hu}-1} \frac{({\hu}-i)}{\tau ^i i!}\right)  d\tau ~=~ \frac{{\hu}^{\hu}}{{\hu}! e^{\hu}}\,.\end{align*}
Then, \ref{lp:spmk}, defined in Theorem \ref{thm:pp_multi},  {is given by} $1 + \ln({\hu}\tau^{\star})$.
\end{oneshot}}

\medskip
\begin{proof}{Proof of Lemma \ref{lem:LP_multi}}
{We  rewrite the second constraint of Problem \eqref{lp:spmk} as
$\int_0^\infty  g_{\hu}(s(\tau)) s(\tau) d\tau ~\le~ 1$, where $g_{\hu}(x)=\frac{f_{\hu}(x)}{x}$. 
 Since by the last set of constraints of Problem \eqref{lp:spmk}, $s(\tau)$ is (weakly) decreasing in $\tau$ and $g_{\hu}(x)$ is decreasing in $x$ (see Lemma \ref{lm:monotonicity} stated at the end of this section),
function $g_{\hu}(s(\tau))$ is increasing in $\tau$. Thus, the optimal solution of Problem \eqref{lp:spmk} must satisfy that $s(\tau) = \min(\hu,1/\tau)$ whenever $\tau \leq \tau^{\star}$ and $s(\tau) = 0$ when $\tau > \tau^{\star}$.
This leads to
\begin{align*}
\int_0^\infty \min(1/\tau, {\hu}) g_{\hu}(\min(1/\tau, {\hu})) d\tau ~&= ~ \int_0^{1/{\hu}} \hu \cdot g_{\hu}({\hu}) d\tau  +\int_{1/{\hu}}^{\tau^{\star}} 
\frac{1}{\tau} g_{\hu}(1/\tau) d\tau \\
~&=~  \frac{{\hu}- e^{-{\hu}} \sum_{i=0}^{{\hu}-1} ({\hu}-i)\frac{{\hu}^i}{i!} }{{\hu}}
+\int_{1/{\hu}}^{\tau^{\star}}  \left({\hu}- e^{-1/\tau} \sum_{i=0}^{{\hu}-1} \frac{({\hu}-i)}{\tau ^i i!}\right)  d\tau ~=~ 1\,.
\end{align*}
Considering that $\frac{\sum_{i=0}^{{\hu}-1} ({\hu}-i)\frac{{\hu}^i}{i!}}{{\hu}}  = \frac{{\hu}^{\hu}}{{\hu}!}$, we have 
\begin{align*}
\int_0^\infty \min(1/\tau, {\hu}) g_{\hu}(\min(1/\tau, {\hu})) d\tau  ~&=~ 1-\frac{{\hu}^{\hu}}{{\hu}! e^{\hu}}
+\int_{1/{\hu}}^{\tau^{\star}}  \left({\hu}- e^{-1/\tau} \sum_{i=0}^{{\hu}-1} \frac{({\hu}-i)}{\tau ^i i!}\right)  d\tau ~=~ 1\,.
\end{align*}
Then, the optimal solution  
of Problem \eqref{lp:spmk}  {is given by} $1 + \ln({\hu}\tau^{\star})$. 
 This is so because 
\[\int_0^{\infty} s(\tau) d\tau ~=~  {\hu}\int_0^{1/{\hu}}  d\tau  +\int_{1/{\hu}}^{\tau^{\star}} 
\frac{1}{\tau}  d\tau  ~= ~1+ \ln({\hu}\tau^{\star})\,.\]

$\blacksquare$}
\end{proof}

{\color{black}\begin{lemma}[$\bf{g_{\hu}(x)}$ Is Monotone]\label{lm:monotonicity}
For any positive integer $\hu$, function $g_{\hu}(x)= \frac{{\hu}- e^{-x} \sum_{i=0}^{{\hu}-1} ({\hu}-i)\frac{x^i}{i!} }{x}$ is decreasing in $x$.
\end{lemma}}

\medskip
\begin{proof}{Proof of Lemma~\ref{lm:monotonicity}}
The plan is to take derivative of $g_{\hu}(x)$ w.r.t. $x$ and show that the derivative is non-positive. Function $g_{\hu}(x)$ is given by
$$g_{\hu}(x)=\frac{{\hu}-e^{-x}\sum_{i =0}^{{\hu}-1} ({\hu}-i)\frac{x^i}{i!}}{x}= \frac{\sum_{i=0}^{{\hu}-1} (1-e^{-x}({\hu}-i)\frac{x^i}{i!})}{x} =\sum_{i=0}^{{\hu}-1} \frac{ 1-e^{-x}({\hu}-i)\frac{x^i}{i!}}{x}\,.$$
By linearity of differentiation, the derivative of $g_{\hu}(\cdot)$ is the sum of derivative of $\frac{ 1-e^{-x}({\hu}-i)\frac{x^i}{i!}}{x}$ for $i=0, \ldots, {\hu}-1$. The derivative for term $i$---that is, $\frac{ 1-e^{-x}({\hu}-i)\frac{x^i}{i!}}{x}$, w.r.t. $x$---is 
$$\frac{\frac{x^i}{i!}e^{-x}({\hu}-i)(x-i)-1+e^{-x}({\hu}-i)\frac{x^i}{i!}}{x^2} ~=~ \frac{\frac{x^i}{i!}e^{-x}({\hu}-i)(x+1-i)-1}{x^2}\,.$$
Therefore, the derivative of $g_{\hu}(x)$ w.r.t. $x$ is given by
$$\frac{\sum_{i=0}^{{\hu}-1}[\frac{x^i}{i!}e^{-x}({\hu}-i)(x+1-i)-1]}{x^2} = \frac{-{\hu}+e^{-x}\sum_{i=0}^{{\hu}-1}[\frac{x^i}{i!}({\hu}-i)(x+1-i)]}{x^2}.$$
The derivative being non-positive is equivalent to 
\begin{align}e^{-x}\sum_{i=0}^{{\hu}-1}\frac{x^i}{i!}({\hu}-i)(x+1-i) \leq {\hu}.\label{eq:der}\end{align}
We divide the sum on the l.h.s. into more manageable terms. Note that
$$\sum_{i=0}^{{\hu}-1}\frac{x^i}{i!}({\hu}-i)(x+1-i)=\sum_{i=0}^{{\hu}-1}\frac{x^i}{i!}{\hu}(x-i)-\sum_{i=0}^{{\hu}-1}\frac{x^i}{i!}\cdot i(x-i) +  \sum_{i=0}^{{\hu}-1}\frac{x^i}{i!}\cdot {\hu} - \sum_{i=0}^{{\hu}-1}\frac{x^i}{i!}\cdot i .$$
We find the value of each of the four terms in the r.h.s. separately. The idea is to take advantage of telescopic sums.
For the first term, we have
\begin{align*}
\sum_{i=0}^{{\hu}-1}\frac{x^i}{i!}{\hu}(x-i)={\hu}\cdot\frac{x^{\hu}}{({\hu}-1)!}\,.
\end{align*}
For the second term, we have
\begin{align*}
\sum_{i=0}^{{\hu}-1}\frac{x^i}{i!}i(x-i)
&=\sum_{j=1}^{{\hu}-1}\sum_{i=j}^{{\hu}-1}\frac{x^i}{i!}(x-i)\\
&=\sum_{i=1}^{{\hu}-1}\frac{x^i}{i!}(x-i)
+\sum_{i=2}^{{\hu}-1}\frac{x^i}{i!}(x-i)
+\ldots
+\sum_{i={\hu}-1}^{{\hu}-1}\frac{x^i}{i!}(x-i)\\
&=\left[ \frac{x^{\hu}}{({\hu}-1)!}-\frac{x}{0!} \right]
+\left[ \frac{x^{\hu}}{({\hu}-1)!}-\frac{x^2}{1!} \right]
+\ldots
+\left[ \frac{x^{\hu}}{({\hu}-1)!}-\frac{x^{{\hu}-1}}{({\hu}-2)!} \right]\\
&=({\hu}-1)\frac{x^{\hu}}{({\hu}-1)!}-\left[ \frac{x^{{\hu}-1}}{({\hu}-2)!}+\ldots+\frac{x^2}{1!}+\frac{x}{0!}\right]\,.
\end{align*}
For the third term, we have
\begin{align*}
\sum_{i=0}^{{\hu}-1}\frac{x^i}{i!}{\hu}
={\hu}\left[ \frac{x^{{\hu}-1}}{({\hu}-1)!} + \frac{x^{{\hu}-2}}{({\hu}-2)!} +\ldots + \frac{x^0}{0!} \right]\,.
\end{align*}
And, finally, for the fourth term we have
\begin{align*}
\sum_{i=0}^{{\hu}-1}\frac{x^i}{i!}i
&=\sum_{i=1}^{{\hu}-1}\frac{x^i}{(i-1)!}=\frac{x^{{\hu}-1}}{({\hu}-2)!} + \frac{x^{{\hu}-2}}{({\hu}-3)!} +\ldots + \frac{x}{0!}\,.
\end{align*}
Putting everything together, we obtain
\begin{align*}
&\sum_{i=0}^{{\hu}-1}\frac{x^i}{i!}({\hu}-i)(x+1-i)\\
&={\hu}\cdot\frac{x^{\hu}}{({\hu}-1)!}-({\hu}-1)\frac{x^{\hu}}{({\hu}-1)!}+\left[ \frac{x^{{\hu}-1}}{({\hu}-2)!}+\ldots+\frac{x^2}{1!}+\frac{x}{0!} \right]\\
&+{\hu}\left[\frac{x^{{\hu}-1}}{({\hu}-1)!} + \frac{x^{{\hu}-2}}{({\hu}-2)!} +\ldots + \frac{x^0}{0!}\right]\\
&-\left[ \frac{x^{{\hu}-1}}{({\hu}-2)!} + \frac{x^{{\hu}-2}}{({\hu}-3)!} +\ldots + \frac{x}{0!} \right]\\
&=\frac{x^{\hu}}{({\hu}-1)!}+{\hu}\left[\frac{x^{{\hu}-1}}{({\hu}-1)!} + \frac{x^{{\hu}-2}}{({\hu}-2)!} +\ldots + \frac{x^0}{0!}\right]\,.
\end{align*}

Note that by the Taylor expansion, 
$e^x = \sum_{i=0}^\infty \frac{x^i}{i!}$. Therefore,
\begin{align*}
e^{-x}\sum_{i=0}^{{\hu}-1}\frac{x^i}{i!}({\hu}-i)(x+1-i)
&=e^{-x} \left[ \frac{x^{\hu}}{({\hu}-1)!} + {\hu}(e^x - \sum_{i={\hu}}^\infty \frac{x^i}{i!}) \right]\\
&=e^{-x} \left[{\hu}\cdot\frac{x^{\hu}}{{\hu}!} + {\hu}(e^x - \sum_{i={\hu}}^\infty \frac{x^i}{i!})\right]=e^{-x} \left[{\hu}(e^x - \sum_{i={\hu}+1}^\infty \frac{x^i}{i!})\right]\leq {\hu}\,. 
\end{align*}
This concludes the proof (see Equation (\ref{eq:der})).
$\blacksquare$ \end{proof}}
\subsection{Proof of  Lemma \ref{lemma:pa_structure}}\label{sec:proof:lem:pos}

{\begin{oneshot}
{\textbf{Lemma~\ref{lemma:pa_structure} (Optimal Mechanism in Position Auction Settings).}}
 For  $j\in [n]$ and $i\in [n]$, let $x_i^j({\bf{v}}) \in \{0,1\}$ and $\pi_i^j(\bf{v}) \in \R_+$ be the
  allocations and payments in the $j$-unit optimal mechanism when buyers' value is ${\bf v}=(v_1, \ldots, v_n)$. Then, the mechanism for the PA
  settings with the following rules is optimal:
  $$x_i({\bf{v}}) = \sum_{j\in [n]} (\alpha_j - \alpha_{j+1}) x_i^j({\bf{v}}) \quad  \text{and} \quad \pi_i({\bf{v}}) = \sum_{j\in [n]} (\alpha_j - \alpha_{j+1}) \pi_i^j(\bf{v})\,,$$  
  where $\alpha_{n+1} = 0$.  
\end{oneshot}}

\medskip
{\begin{proof}{Proof of Lemma \ref{lemma:pa_structure}}
As stated earlier, for any buyer $i$, there is $J_i$ such that
  $x_i^{j}({\bf{v}}) = 0$ for $1 \leq j < J_i$ and $x_i^{j}({\bf{v}}) = 1$ for $J_i \leq j \leq n$. (When buyer $i$ is not allocated in any of the $n$ multi-unit auctions, we set $J_i$ to $n+1$.) In this case,  $\sum_{j\in [n]} (\alpha_j - \alpha_{j+1}) x_i^j({\bf{v}}) = \alpha_{J_i}$, which is the click-through-rate of position $J_i$.   
Now, consider the optimal mechanism in the PA setting. 
In this auction, positions are assigned in a decreasing order of buyers' (ironed) virtual values---that is, the first position is allocated to the buyer with the highest non-negative virtual value, the second position is allocated to the buyer with the second highest non-negative (ironed) virtual value, and so on. This implies that in the optimal mechanism, position $J_i$ must be allocated to buyer $i$, as  buyer $i$  has the $J_i$-th highest virtual value. With regard to the payment, we note that the payment rule in the optimal Myersonian mechanism is a linear function of the allocation rule. Considering this and the fact that  $x_i^j({\bf{v}})$ results in payment of $\pi_i^j({\bf{v}})$, then $\sum_{j\in [n]} (\alpha_j - \alpha_{j+1}) x_i^j({\bf{v}})$ results in an expected payment of $\sum_{j\in [n]} (\alpha_j - \alpha_{j+1}) \pi_i^j(\bf{v})$ for buyer $i$. \
$\blacksquare$\end{proof}}

\subsection{Proof of Theorem  \ref{thm:pp2}}\label{sec:proof:claim}
\begin{oneshot}
{\textbf{Theorem~\ref{thm:pp2} (Revenue Bound of ESP Auctions in Single-unit Settings).}}
	In a single-unit $n$-buyer setting with independent private values, there {exists a vector of prices} $\mathtt{\bp} =(\mathtt{p}_1, \mathtt{p}_2, \ldots, \mathtt{p}_n)$ such that \begin{itemize}\item \textbf{{Non-discretized Bound.}} $\EG(\mathtt{\bp}) ~\geq~
	 \Opt \cdot \frac{1}{{\text{\sf 
FR-ESP}}}$, and 
\item \textbf{{Discretized Bound.}} $\EG(\mathtt{\bp}) ~\geq~\Opt \cdot \frac{1}{{\text{\sf \small{FR-ESP-d}}}(k)}$ for any positive integer $k$,
\end{itemize}
where $\Opt$ is the expected optimal revenue, $\EG(\mathtt{\bp})$ is the expected revenue of an ESP auction with personalized reserve prices $\mathtt{\bp}$, and $\text{\sf 
FR-ESP}$ and ${{\text{\sf \small{FR-ESP-d}}}(k)}$ are defined as 
 \begin{multicols}{2} \setlength{\columnseprule}{1pt}
\noindent
\begin{align*} 
&\text{\sf{FR-ESP}} ~=~\\ &{\max_{\{{s(\tau), \tau\ge 0\}}}} ~\int_0^\infty 
s(\tau) d\tau\quad \text{s.t.} \nonumber\\[8pt]
&  \begin{aligned}
 &\int_{\T_x}^\infty  (2-2e^{-s(\tau)} - s(\tau)e^{-s(\tau)}) d\tau \\[8pt]
 &~~+\int_0^{\T_x} ( x +  (1-e^{-s(\tau)})) d\tau \le2~~~~  \forall x\in [0,1] \nonumber \\
& \int_0^\infty  f(s(\tau)) d\tau ~\le~ 1\\[8pt]
&\text{$s(\cdot)$ is weakly decreasing}\,,~~&\nonumber
\end{aligned}
 \label{lp:esp} \tag{\small{\sf{FR-ESP}}}
\end{align*}
\begin{align*}
&{\text{\sf \small{FR-ESP-d}}}(k)= \\
& \max_w \sum_{i\in [k]} w_i \quad \text{s.t.} \nonumber \\
&\begin{aligned}   &\sum_{i=1}^j  w_i 
\frac{2(1-e^{-\mathtt{s}_i}) -
 \mathtt{s}_ie^{-\mathtt{s}_i}}{\mathtt{s}_i} \\ 
\quad &~~~+ \sum_{i=j+1}^{k} w_i \frac{ \mathtt{s}_j+(1-e^{-\mathtt{s}_i}) }{\mathtt{\mathtt{s}}_i}
    ~\leq~ 2,  &\forall j \in [k] \\
   & \sum_{i\in [k]} w_i \frac{1-e^{-\mathtt{s}_i}}{\mathtt{s}_i}~ \leq ~ 1 & \\
   & w_i \geq 0\,. &\forall i \in [k]
\end{aligned}
\label{lp:esp:d} \tag{\small{\sf FR-ESP-d}}
\end{align*}
\end{multicols}
Here, $f(x) = ({1 - e^{-x}})$, for any $x\in [0,1]$,  $\T_x~=~  
\inf\{\tau: s^\star(\tau) \le  x\}$, and  $\mathtt{s}_i = i/k$, for $i\in [k]$. 
Further, setting $k=3200$, the approximation factor is 
$\frac{1}{\text{\small{\sf FR-ESP-d}}(3200)} = 0.6620$. 
\end{oneshot}
\begin{proof}{Proof of Theorem  \ref{thm:pp2}} 
The proof has two parts. First, we show the non-discretized bound and then we verify the discretized one.

\textbf{Non-discretized Bound.} The proof is similar to the proof of Theorem \ref{thm:pp1}. The main difference is showing that $s^\star(\cdot)$---corresponding to the optimal mechanism---satisfies the first set of constraints in Problem \eqref{lp:esp}. (Note that the only thing that differentiates Problems \eqref{lp:esp} and \eqref{lp:spm} is their first sets of constraints.) Thus, here, we only focus on the main difference and exclude the remainder of the proof. 

We begin with a few definitions.

\textbf{Myersonian ESP Auction.}  We run the $\EG$ auction with personalized reserve prices for each buyer, with buyer $i$ facing the re-sampled threshold $t_i'$ as his reserve price (see the definition in 
Section~\ref{prelim:defn}). Let $\ME$ denote the expected revenue of this auction.

\textbf{Uniform ESP Auction.} We run the $\EG$ auction with a uniform reserve price of $p^{\star}_E = \argmax_p \E_{\mathbf{v}}\left[\max(p, 
v_{(2)})\cdot\ind(\max_{i\in [n]} 
v_i \geq p)\right]$, 
where $v_{(2)}$ is the second-highest bid (which is also equal to the second-highest value in a 
truthful auction). We denote the revenue of this auction by $\UE$.

Now, we show that $s^{\star}(\cdot)$ satisfies the first set of constraints. 

\textbf{First Set of Constraints.} 
Let $\T_x~=~  \inf\{\tau: 
s^\star(\tau) \le  x\}$, $x\in [0,1]$.  
In addition, with a slight abuse of notation, let $\UE_x$ be the revenue of the ESP auction that posts a 
uniform price of $\T_x$ for all buyers. By the definition of the uniform ESP auction, we obtain $\UE~=~ \max_{x\in 
[0,1]} ~ \UE_x$.  We now bound $\ME+\UE_x$ for any $x\in [0,1]$. As usual, without loss of generality, we assume that  $\max(\ME, \UE)=1$.

 We begin by bounding $\UE_x$. 
 Define $u_x(\tau)$ as the probability that the ESP auction with the uniform price $\T_x$ sells with a price of at 
least
$\tau$. Then, $\UE_x = \int_{\tau =0}^{\infty} u_x(\tau) d\tau$.  
Next, we bound $\UE_x$ by bounding $u_x(\tau)$. As we argued in the proof of Theorem~\ref{thm:pp_n}, for any $\tau \le \T_x$, we have $u_x(\tau)$ by 
$u_x(\tau) ~ \geq~ s^\star(\T_x)~ \geq ~ x$ for $\tau \le \T_x$; see Equation \eqref{eq:u_x_1_1}. 
For $\tau >\T_x$, we bound $u_x(\tau)$ by noting that the  ESP auction with uniform price $\T_x$ can sell at a price of at least $\tau$ only if there are at
least two buyers bidding above $\tau$.  Let $\widehat Z_{\tau}=  \sum_{i=1}^n\ind (v_i 
~\geq ~ {\tau})$ and $Z_{\tau}= \sum_{i=1}^n\ind (v_i 
~\geq~  t_i' ~\geq ~ {\tau})$.

Then, we have
\begin{align}u_x(\tau) ~=~ \P[\widehat Z_{\tau} \geq 2] ~\geq~ \P[Z_{\tau} \geq 
2] ~=~ 1 - \P[Z_{\tau} = 0] - 
\P[Z_{\tau} = 1]\,, \quad \tau >\T_x\,. \label{eq:u_x_2}\end{align}
Combining these two bounds, we obtain
\begin{align} \label{eq:uniform_eager}
\UE_x ~=~ \int_{\tau =0}^{\infty} u_x(\tau) d\tau ~\ge ~ \int_0^{\T_x} x d\tau + \int_{\T_x}^\infty (1 -  \P[Z_{\tau} = 0] - \P[Z_{\tau} = 1]) d\tau\,.
\end{align}

We next bound $\ME$.  With a slight abuse of notation, let $\m(\tau)$ be the probability that the Myersonian ESP auction sells at a price greater than or equal to $\tau$. 
Then, by the definition of $Z_{\tau}$, we have 
\begin{align}\label{eq:m_bound}\m(\tau)~ =~ \P[Z_{\tau} \geq 1] ~ \geq 
~1-\P[Z_{\tau} = 0]\,.
\end{align}
Then, considering  that  $\ME = \int_{\tau =0}^{\infty} \m(\tau) d\tau 
$ and by using  {Inequalities  (\ref{eq:uniform_eager}) and (\ref{eq:m_bound})}, we obtain
\begin{align} \label{eq:bound_sum}
2 ~\ge ~\UE_x+ \ME ~\ge ~  \int_0^{\T_x} \big(x+1-\P[Z_{\tau} = 0]\big) d\tau + \int_{\T_x}^\infty (2 - 2 \P[Z_{\tau} = 0] - \P[Z_{\tau} = 1]) d\tau\,, 
\end{align}
where the first inequality follows from our assumption that $\max(\ME, \UE) 
=1$. To simplify the r.h.s. of ~\eqref{eq:bound_sum}, we utilize 
Lemma~\ref{lem:prob}, which says $\P[Z_{\tau} = 0] ~\le~ e^{-s^\star(\tau)}$ and  
$2\P[Z_{\tau} = 0] + \P[Z_{\tau} = 1] ~\le~ (2 + s^\star(\tau)) e^{-s^\star(\tau)}$. This yields
 \begin{align}  \int_0^{\T_x} ( x +  (1-e^{-s^\star(\tau)})) d\tau + \int_{\T_x}^\infty  (2-2e^{-s^\star(\tau)} - s^\star(\tau)e^{-s^\star(\tau)}) d\tau~\le ~ 2\,. \label{eq:ineq_UE_ME}
  \end{align}
  The above equation holds for any $x \in [0,1]$ and it confirms that $s^\star(\cdot)$ satisfies the first set of constraints of Problem \eqref{lp:esp}.
  
  \textbf{Discretized Bound.} Thus far, we showed the non-discretized bound. The proof of the discretized bound here is similar to the proof of the discretized bound in Theorem \ref{thm:pp_n}. However, showing that the $w_i^\star$'s associated with the optimal mechanism satisfy the first set of constraints of Problem \eqref{lp:esp:d} does not directly follow from the proof of Theorem \ref{thm:pp_n}. (This is not the case for its second constraint.) Thus, here we focus on the first set of constraints and exclude the remainder of the proof.

As we already established in the first part of the proof that, for any $x\in [0,1]$, we have
 \begin{align*} 2 ~\ge~ \int_0^{\T_x} ( x +  (1-e^{-s^\star(\tau)})) d\tau + \int_{\T_x}^\infty  (2-2e^{-s^\star(\tau)} - s^\star(\tau)e^{-s^\star(\tau)}) d\tau\,,
  \end{align*}
  where $\T_x=\inf\{\tau: 
s^\star(\tau) \le  x\}$, $x\in [0,1]$. Set $x = j/k$. Then, we have
\begin{align}
2 &~\ge~  \int_0^{\T_x} ( x +  (1-e^{-s^\star(\tau)})) d\tau + \int_{\T_x}^\infty  (2-2e^{-s^\star(\tau)} - s^\star(\tau)e^{-s^\star(\tau)}) d\tau \nonumber \\
  &~=~ \sum_{i=j+1}^{k} \int_{\uptau_{i}}^{\uptau_{i-1}}  \frac{( x +  
  (1-e^{-s^\star(\tau)}))}{s^\star(\tau)} s^\star(\tau) d\tau +  \sum_{i=1}^j 
  \int_{\tau_{i}}^{\tau_{i-1}} \frac{(2-2e^{-s^\star(\tau)} - 
  s^\star(\tau)e^{-s^\star(\tau)})}{s^\star(\tau)}
  s^\star(\tau) d\tau \nonumber \\ 
  &{~\ge~    \sum_{i=j+1}^{k} w_{i}^\star \frac{ x + (1-e^{-\mathtt{s}_i})}{\mathtt{s}_i} + \sum_{i=1}^j  w_{i}^\star 
  \frac{(2-2e^{-\mathtt{s}_i} - \mathtt{s}_i e^{-\mathtt{s}_i})}{\mathtt{s}_i}\,,}   \label{eq:second_const}
\end{align}
where the  equality follows from the definitions of $\uptau_j$'s and $\T_x$ and the 
fact that at $x =j/k$,  $\T_x = \uptau_j$.
Recall that $0 = \uptau_k \leq \uptau_{k-1} \leq \ldots \leq \uptau_1 \leq
\uptau_0 = \infty$ such that $\uptau_j =\inf\{\tau: s^\star(\tau) \le  j/k\}$, $j\in [k-1]$, and
$\Opt ~=~ \sum_{i\in [k]} w_{i}^\star$, where $w_{i}^\star ~=~ \int_{\uptau_{i}}^{\uptau_{i-1}} s^\star(\tau) d\tau$. 
The second inequality follows from the definitions of $w_{i}^\star$'s, and $\mathtt{s}_i$'s 
and the facts that $y\mapsto \frac{1}{y}(2-(2+y)e^{-y})$ and $y\mapsto 
\frac{1}{y} (x + 1-e^{-y}) $ are decreasing in $y \in [0,1]$ (for proof, see Lemma \ref{lem:monotone_2}) and that $s^\star(\tau)$ itself is a decreasing function.   
$\blacksquare$
\end{proof}

\begin{lemma} \label{lem:monotone_2}
Functions $y\mapsto \frac{1}{y}(2-(2+y)e^{-y})$ and $y\mapsto \frac{1}{y} (x + 
1-e^{-y}) $ are decreasing in $y \in [0,1]$  and $y \in
  [x,1]$. Further, function
  $y \mapsto \frac{1}{y}(2-\p_n(y))$ 
  is decreasing\footnote{The monotonicity of $\frac{1}{y}(2-\p_n(y))$ in $y$ is used in the proof of Theorem \ref{thm:esp_n}, which is the n-dependent version of Theorem \ref{thm:pp2}. This proof of Theorem~\ref{thm:esp_n} is omitted because of its similarity to the proof for Theorem~\ref{thm:pp2}.} in $y \in [0,1]$ for every positive integer $n$ and every $x \geq 0$, where 
   $\p_n(y) = 2(1-\frac{y}{n})^n + y(1-\frac{y}{n})^{n-1}$.
\end{lemma}
\begin{proof}{Proof of Lemma \ref{lem:monotone_2}}
For the first function, note that $\frac{d\left(
     \frac{1}{y}(2-(2+y)e^{-y}) \right) }{dy} = \frac{1}{y^2}(2ye^{-y}+y^2e^{-y}+2e^{-y}-2)
   \leq 0 $ due to the inequality $e^y \ge 1+y+\frac{y^2}{2}$. For the second
   function, \[\frac{\partial \left(  \frac{1}{y} (x + 1-e^{-y}) \right)}{\partial y}  =
   \frac{ye^{-y}-x-1+e^{-y}}{y^2} \le 0\,,\quad \]
    where the inequality holds because $1+y \leq e^y$ and $x \geq 0$. 
    
  Next, we show that function
  $y \mapsto \frac{1}{y}(2-\p_n(y))$  
  is decreasing in $y \in [0,1]$ for every positive integer $n$ and every $x \geq 0$.
 By definition of $\p_n(\cdot)$, we have 
\begin{align}\frac{d \left(\frac{1}{y}(2-\p_n(y))\right)}{dy} ~=~ 
\frac{(1-\frac{y}{n})^{n-2}\left(y^2\frac{n-1}{n}+2y(1-\frac{y}{n})+2(1-\frac{y}{n})^2-\frac{2}{(1-\frac{y}{n})^{n-2}}\right)}{y^2}\,.
\end{align}
Note that the derivative is non-positive if 
\begin{align}
y^2\frac{n-1}{2n}+y(1-\frac{y}{n})+(1-\frac{y}{n})^2 ~=~ 1+\frac{n-2}{n}y+\frac{(n-1)(n-2)}{2}\frac{y^2}{n^2}~\le~ \frac{1}{(1-\frac{y}{n})^{n-2}}\,,
\end{align}
where the equality follows from simple algebra. Below, we verify the inequality. This reveals that $\frac{1}{y}(2-\p_n(y))$ is decreasing in $y$. 
For any $y\in[0,1)$, we have
\begin{align*}
\dfrac{1}{(1-\frac{y}{n})^{n-2}}&~\ge~ \left( 1+\frac{y}{n}+\frac{y^2}{n^2}  \right)^{n-2}\\
&~\ge~ 1+(n-2)\frac{y}{n}+\left( \frac{(n-2)(n-3)}{2}+ (n-2) \right)\frac{y^2}{n^2}\\
&~=~1+\frac{n-2}{n}y+\frac{(n-1)(n-2)}{2}\frac{y^2}{n^2}\,.
\end{align*} 
The last equation is the desired result.
    
$\blacksquare$
\end{proof}
 \footnotesize{
\bibliographystyle{informs2014}
\bibliography{sigproc.bib}}
\end{document}